%% file: main.tex
\documentclass[11pt,letterpaper]{article}

\usepackage[margin=1in]{geometry}
\usepackage[dvipsnames]{xcolor}
\usepackage{latexsym,graphicx,amssymb}
\usepackage{amsmath}
\usepackage{amsthm, amsfonts}
\usepackage{mathtools, thmtools}
\usepackage{bbm}
\usepackage{float}
\usepackage{xspace}
\usepackage{paralist}
\usepackage{enumerate,multicol}
\usepackage{cases}
\usepackage{caption}
\usepackage{graphicx}
\usepackage{tikz,pgfmath}
\usepackage{pdfcomment}
\usetikzlibrary{matrix,positioning,quotes}
\usetikzlibrary{shapes.multipart}
\usetikzlibrary{calc}
\usetikzlibrary{arrows,decorations.markings}
\usetikzlibrary{matrix}
\usepackage{url}
\usepackage[ruled,linesnumbered,vlined]{algorithm2e}
\usepackage{hyperref}
\hypersetup{hidelinks}
\usepackage[capitalise]{cleveref}
\crefname{Distribution}{Distribution}{Distributions}

\usepackage[noend]{algpseudocode}
\usepackage{multirow}
\usepackage{graphicx}
\usepackage{subcaption}
\usepackage{varwidth}
\usepackage{wrapfig}
\usepackage{enumitem}
\usepackage{float}
\usepackage{bbm,bm}
\usepackage{tcolorbox}

\usepackage[d]{esvect}
\usepackage{bm}
\usepackage{mleftright}
\usepackage{ulem}
\usepackage{stackengine}
\usepackage{mathrsfs}
\usepackage{bbm}
\usepackage{soul}
\usepackage[numbers]{natbib}

\newenvironment{proofof}[1]{\begin{proof}[Proof of #1]}{\end{proof}}
\numberwithin{figure}{section}

\newtheorem{theorem}{Theorem}[section]
\newtheorem{definition}[theorem]{Definition}

\newtheorem{remark}[theorem]{Remark}

\newtheorem{lemma}[theorem]{Lemma}
\newtheorem{claim}[theorem]{Claim}
\newtheorem*{theorem*}{Theorem}
\newtheorem*{lemma*}{Lemma}

\input{head}

\input{notations}

\author{
  Tianxiao Li \thanks{Institute for Interdisciplinary Information Sciences, Tsinghua University. \url{litx20@mails.tsinghua.edu.cn}.}
  \and
  Jingxun Liang \thanks{Institute for Interdisciplinary Information Sciences, Tsinghua University. \url{liangjx20@mails.tsinghua.edu.cn}.} 
  \and
  Huacheng Yu \thanks{Department of Computer Science, Princeton University. \url{yuhch123@gmail.com}.}
  \and
  Renfei Zhou \thanks{Institute for Interdisciplinary Information Sciences, Tsinghua University. \url{zhourf20@mails.tsinghua.edu.cn}.}
}
\title{Dynamic ``Succincter''}
\date{}

\begin{document}
\maketitle
\thispagestyle{empty}
\setcounter{page}{0}
\begin{abstract}
  Augmented B-trees (aB-trees) are a broad class of data structures.
  The seminal work ``succincter'' by P\v{a}tra\c{s}cu~\cite{patrascu2008succincter} showed that any aB-tree can be stored using only two bits of redundancy, while supporting queries to the tree in time proportional to its depth.
  It has been a versatile building block for constructing succinct data structures, including rank/select data structures, dictionaries, locally decodable arithmetic coding, storing balanced parenthesis, etc.

  In this paper, we show how to ``dynamize'' an aB-tree.
  Our main result is the design of dynamic aB-trees (daB-trees) with branching factor two using only three bits of redundancy (with the help of lookup tables that are of negligible size in applications), while supporting updates and queries in time polynomial in its depth.
  As an application, we present a \emph{dynamic} rank/select data structure for $n$-bit arrays, also known as a dynamic \emph{fully indexable dictionary} (FID)~\cite{RRR02}.
  It supports updates and queries in $O(\log n/\log\log n)$ time, and when the array has $m$ ones, the data structure occupies
  \[
    \log\binom{n}{m}+O(n/2^{\log^{0.199} n})
  \]
  bits.
  Note that the update and query times are optimal even without space constraints due to a lower bound by Fredman and Saks~\cite{FS89}.
  Prior to our work, no dynamic FID with near-optimal update and query times and redundancy $o(n/\log n)$ was known.
  We further show that a dynamic sequence supporting insertions, deletions and rank/select queries can be maintained in (optimal) $O(\log n/\log\log n)$ time and with $O(n\cdot \poly\log\log n/\log^2 n)$ bits of redundancy.
\end{abstract}

\newpage

\section{Introduction}
\input{intro}

\subsection{Technical Contribution}\label{sec:tech}
\input{tech}

\subsection{An Overview of Augmented B-trees}\label{sec:abtree}
\input{prelim}

\subsection{Organization}
\input{org}

\section{Adapters}
\label{sec:adapter}
\input{adapter}

\section{Dynamic aB-trees}
\label{sec:daB}
\input{dyn_aBtree}

\section{Removing the Assumption}
\label{sec:remove_assumption}
\input{remove_assumption}

\section{Applications}
\label{sec:application}
\input{application}

\section{Lower Bounds for Two-Way Adapters}
\label{sec:lb}
\input{2way_adapter_lb}

\bibliographystyle{alpha}
\bibliography{reference.bib}

\end{document}

%% file: notations.tex
\renewcommand{\left}{\mleft}
\renewcommand{\right}{\mright}

\newcommand{\ceil}[1]{\left\lceil #1 \right\rceil}

\newcommand{\bigceil}[1]{\bigl\lceil #1 \bigr\rceil}

\newcommand{\floor}[1]{\left\lfloor #1 \right\rfloor}
\newcommand{\midfloor}[1]{\lfloor #1 \rfloor}
\newcommand{\bigfloor}[1]{\bigl\lfloor #1 \bigr\rfloor}

\newcommand{\bk}[1]{\left( #1 \right)}
\newcommand{\midbk}[1]{( #1 )}
\newcommand{\bigbk}[1]{\bigl( #1 \bigr)}
\newcommand{\Bigbk}[1]{\Bigl( #1 \Bigr)}

\newcommand{\Bk}[1]{\left[ #1 \right]}
\newcommand{\midBk}[1]{[ #1 ]}
\newcommand{\bigBk}[1]{\bigl[ #1 \bigr]}

\newcommand{\BK}[1]{\left\{ #1 \right\}}
\newcommand{\midBK}[1]{\{ #1 \}}

\newcommand{\abs}[1]{\left\lvert #1 \right\rvert}
\newcommand{\midabs}[1]{\lvert #1 \rvert}

\DeclareMathOperator*{\E}{\mathbb{E}}

\DeclareMathOperator{\poly}{poly}

\renewcommand{\tilde}{\widetilde}
\newcommand{\eqdef}{\eqqcolon}
\newcommand{\defeq}{\coloneqq}
\newcommand{\eps}{\varepsilon}

\renewcommand{\l}{\ell}
\renewcommand{\emptyset}{\varnothing}
\renewcommand{\epsilon}{\eps}

\renewcommand{\emptyset}{\varnothing}

\newcommand{\smallsub}{\scriptscriptstyle}

\newcommand{\wordlen}{w}

\newcommand{\patrascu}{\textup{P{\v{a}}tra{\c{s}}cu}}

\newcommand{\Lmax}{L_{\max}}
\newcommand{\incword}{\tilde{m}}

\newcommand{\Nchunk}{N_{\textup{chk}}}

\newcommand{\I}{\mathcal{I}}

\newcommand{\cost}{\textup{\textsf{cost}}}

\newcommand{\dist}{\textup{dist}}

\newcommand{\treenum}[1][]{\mathcal{N}\bk{#1}}
\newcommand{\labfunc}{\mathcal{A}}
\renewcommand{\phi}{\varphi}

\newcommand{\Mcat}{M_{\textup{cat}}}
\newcommand{\mcat}{m_{\textup{cat}}}

\newcommand{\Mmax}{M_{\textup{max}}}

\newcommand{\kcat}{k_{\textup{cat}}}

\newcommand{\mmleft}{m_{\textup{fix}}}
\newcommand{\mmright}{m_{\textup{rem}}}

\newcommand{\Mleft}{M_{\textup{fix}}}
\newcommand{\Mright}{M_{\textup{rem}}}
\let\Mfix\Mleft

\newcommand{\ym}{y_{\smallsub M}}
\newcommand{\yk}{y_{\smallsub K}}

\newcommand{\curn}{n_u}

\newcommand{\planM}{\tilde{M}}

\newcommand{\X}{\mathcal{X}}

\newcommand{\mychoose}[2]{{{#1}\atop#2}}

\newcommand{\Bigbinom}[2]{\Bigbk{\mychoose{#1}{#2}}}

%% file: intro.tex
Succinct data structures~\cite{jacobson1988succinct} are data structures that use space very close to the information-theoretical optimum.
To store data of size $n$ bits, a succinct data structure uses $n+o(n)$ bits of space, where the $o(n)$ term is referred to as the redundancy, while supporting operations efficiently.
Despite such strong requirements, efficient succinct data structures have been proposed for many fundamental problems, including dictionaries and filters~\cite{BM99,Pagh01,RRR02,RR03,patrascu2008succincter,ANS10,PSW13,Yu20,LYY20,Bercea2020ADS,BFKKL22,bender2023tiny}, rank and select data structures~\cite{jacobson1988succinct,Jacobson89,CM96,Munro96,MRR01,RRR02,GGGRR07,GRR08,patrascu2008succincter,Yu19}, storing trees and strings~\cite{hon2003succinct,chan2004compressed,makinen2008dynamic,chan2007compressed,gonzalez2009rank,he2010succinct,NS14,NN14,bender2023tiny}, initializable arrays~\cite{HK17,katoh2022place}, etc.

Many of these succinct data structures are static, i.e., the data is fixed and given in advance, then it is preprocessed into a data structure supporting fast queries.
The seminal paper ``succincter'' by \patrascu{}~\cite{patrascu2008succincter} proposed generic techniques for constructing static succinct data structures.
In particular, it was shown that any \emph{augmented B-tree} (aB-tree) can be compressed with only \emph{two} bits of redundancy, while supporting efficient queries.
Augmented B-trees are a class of generic tree data structures, which turn out to be applicable to several central problems in succinct data structures, including dictionaries, rank/select data structures, balanced parentheses matching, etc.
Using this compression of aB-trees, one is able to design succinct data structures with $O(n/\log^c n)$ bits of redundancy for these problems with constant query time for any constant $c$, while most prior techniques can only give a redundancy of around $O(n/\log n)$.\footnote{For certain problems including rank/select, this was a ``formal barrier'' to some extent, as there is a matching lower bound against any \emph{systematic encoding}~\cite{Golynski07}.}

On the other hand, it is a more challenging task to dynamically maintain data under updates within succinct space, while supporting efficient queries.
Despite the versatility of \patrascu{}'s succinct aB-trees, it was not known how to efficiently update a general succinct aB-tree.
In fact, as we will discuss in \cref{sec:tech}, there is a common difficulty to dynamize any succinct data structure with variable-length components.

In this paper, we propose a generic technique to store multiple variable-length data structures supporting efficient updates, and apply this technique to design new dynamic succinct data structures.

\subsection{Our results}
Our main result shows that a dynamic augmented B-tree (daB-tree) with branching factor two can be stored with three bits of redundancy, and can be updated and queried in time polynomial in the \emph{height} of the tree.
Dynamic augmented B-trees are formally defined in \cref{def:daB_tree} (see also \cref{sec:abtree}).
At a high level, a daB-tree is a tree data structure maintaining a dynamic array $A[1..n]$, such that every leaf corresponds to an entry $A[i]$, and every tree node is associated with a label.
The label of leaf $i$ is a function of $A[i]$, and the label of an internal node is a function of the labels of its children.
Thus, the entire tree is determined by the underlying array $A$, while the labels are designed such that they facilitate the queries on $A$.
In this paper, we will focus on maintaining binary daB-trees, although our technique also applies to general branching factor $B>2$ with a worse bound.

As a concrete example, for $A\in\{0,1\}^n$, if we want to design a data structure that can return $A[1]+\cdots+A[i]$ efficiently for any given $i$ (a.k.a.~the rank queries), we can set the label of an internal node to be the sum of its children, and set the label of a leaf to be the value of $A[i]$.
Thus, the label of any internal node is the sum of its leaves.
Given access to the labels, one can compute any partial sum in time $O(\log n)$.

We prove that \emph{any} such tree (defined by the functions computing the labels) can be stored with only three bits of redundancy \emph{conditioned on} the root label,\footnote{The space benchmark is to store one of the possible arrays with this particular root label. See also \cref{sec:abtree}.} such that both recovering the label of any internal node and updating a single entry $A[i]$ take time polynomial in the height of the tree.

\begin{theorem}[Informal version of \cref{clm:dyn_aBtree}]
	A daB-tree maintaining an array $A\in\Sigma^{n}$ with label set $\Phi$ can be stored with three bits of redundancy such that an update or a query takes $\poly\log n$ time, as long as the word-size $w\geq \Omega(\log (n\cdot\left|\Sigma\right|\cdot\left|\Phi\right|))$, assuming the access to a fixed lookup table of size $\poly(n, |\Sigma|, |\Phi|)$.
\end{theorem}

We remark here that in most applications, we will divide the input data into blocks of poly-log sizes, and maintain a daB-tree for each block. Thus, the update and query time of a daB-tree is in fact poly-log-log in the data size.\footnote{It is also worth noting that our design of daB-tree is \emph{strongly history-independent}, i.e., the encoding of a daB-tree only depends on the array it stores, but not the operations in the past. However, none of the applications listed below in this paper are history-independent, because we use history-dependent subroutines to manage memory chunks (\cref{lm:chunking}).}

Using daB-trees, we obtain improved dynamic succinct data structures for several problems.

\paragraph{Rank/select.}
In \textsc{Rank/Select} problem with modifications, we need to maintain an array $A[1..n]$ with $m$ ones and $n-m$ zeros. Succinctly storing this array requires $\log\binom{n}{m}$ bits of space. In each operation performing on $A$, we may modify $A[k]$, query the number of ones in $A[1..k]$ (\textsc{Rank} queries), or query the index of the $k$-th one (\textsc{Select} queries) for any $k$. 
Such a data structure is also called a (dynamic) fully indexable dictionary~\cite{RRR02}.

By using daB-trees to maintain $A$, we can solve \textsc{Rank/Select} problem with $O(n/2^{(\log n)^{1/5-o(1)}})$ bits of redundancy and $O(\log n/\log \log n)$ query/update time. Even with no space constraint, our time already matches the lower bound~\cite{FS89,patrascu2014dynamic}.

\paragraph{Locally-decodable arithmetic coding with small alphabet size.}
The arithmetic coding problem is defined on an array $A[1..n]$ with alphabet $\Sigma$ satisfying $|\Sigma|=O(1)$. For $\sigma\in\Sigma$, denote the number of occurrences of $\sigma$ in $A$ by $f_\sigma$, then storing $A$ requires at least $\log \binom{n}{f_1 f_2 \cdots f_{|\Sigma|}}$ bits of space (the problem setting assumes $\{f_\sigma\}$ are stored outside and are not counted in the space usage). We need to support modifications and queries on any single $A[i]$. Via our method, we can achieve redundancy of $O(n/\poly\log n)$ bits for arbitrary $\poly\log n$ on the denominator, with $O(\log^2 \log n)$ time for each query, and $O(\log^5 \log n)$ time for each modification.

\paragraph{Dynamic sequences.}
We further show that we can maintain a dynamic sequence over an alphabet $\left|\Sigma\right|=O(1)$ under insertions and deletions of the symbols, supporting \textsc{Rank/Select} queries (for the case $|\Sigma|>2$, \textsc{Rank}$(k,\sigma)$ is defined as the number of occurrences of $\sigma$ in the first $k$ entries, and \textsc{Select}$(k,\sigma)$ is the location of the $k$-th $\sigma$).
In our result, we present a data structure using $O(n \log^9\log n/\log^2 n)$ bits of redundancy with $O(\log n/\log \log n)$ time for any operation. This improves the previous best-known redundancy of $O(n/\log ^{1-\epsilon} n)$ with optimal time~\cite{NN14}.

\paragraph{Improved range min-max trees.}
Dynamic sequences have applications to more problems such as maintaining dynamic trees introduced in \cite{NS14}. In this problem, we use a parenthesis sequence $P[1,\dots,2n]$ to represent the dynamic rooted tree with $n$ nodes, and support various operations such as finding the matching parenthesis of $P[i]$, inserting/deleting a pair of parenthesis, performing \textsc{Rank/Select} on opening or closing parenthesis, and so on.

For all these operations, \cite{NS14} uses a data structure called range min-max tree to solve this problem with $O(\log n/\log \log n)$ time for most operations. By using daB-trees, we get the same time complexity, but with a better redundancy of $O(n \log^9\log n/\log^2 n)$ bits.

%% file: tech.tex
In the design of many static succinct data structures, the main data structure is a concatenation of multiple components.
The preprocessing algorithm first generates each component based on the data, then concatenates them with possibly another (small) structure that navigates the query algorithm to each component.
In order for the whole data structure to be succinct, each component needs to be preprocessed succinctly respectively.
Often this leads to the variable lengths of the components, i.e., their lengths may depend on the input.
In particular, the aB-tree recursively combines $B$ smaller potentially variable-length data structures that also have this structure.
Updates can be very challenging for such data structures.
As we update the data, the length of a component may change.
Hence, if we naively concatenate the components, when one component increases its length, we will have to shift all subsequent ones to make room for it, which is unaffordable.

One way to dynamize such a data structure is to apply the ideas of~\cite{BCDMS99}, which works well when each component is not too small.
Roughly speaking, if each component has size $\approx L$, then we divide memory into blocks of size $\sqrt{L}$ such that each component occupies an integer number of blocks.
For every component, we maintain an array of pointers pointing to the list of memory blocks it is currently using.
If a component increases its size and needs more space, we allocate a whole block for it; if a component shrinks its size, we release its unused blocks.
In this way, each component can still be accessed as usual.
The redundancy becomes $O\bigbk{\sqrt{L}\log n}$ bits per component, i.e., roughly $O\bigbk{(\log n)/\sqrt{L}}$-fraction of the memory is the redundancy.
Hence, when $L$ is at least (a large) $\poly\log n$, this approach gives small redundancy.

However, in many constructions (including the rank/select data structures, succinct aB-trees, dynamic sequences, etc), the basic components have smaller sizes.
Our main technical contribution is a generic way to jointly store two variable-length (and potentially small) data structures with \emph{no} redundancy, such that each of them can be accessed and updated (changing sizes) efficiently.
Consider two data structures of $\l_1$ words and $\l_2$ words respectively.
We would like to store them in $L=\l_1+\l_2$ consecutive words with \emph{no} auxiliary data.
Moreover, they need to be accessed and updated efficiently:
\begin{enumerate}[label=(\alph*)]
\item given $(i, j)$ for $i\in\{1,2\}$ and $j\in[\l_i]$, find efficiently where the $j$-th word in the $i$-th data structure is stored;
\item support the allocation or release of a word for one of the two components, i.e., increment or decrement $\l_i$ (and consequently, $L$).
\end{enumerate}
Naively concatenating two data structures gives fast access (item (a)), while allocation or release can take $\Theta(L)$ time.
Storing the first data structure in odd addresses and the second in the even addresses can give fast access, allocation and release, but may induce $\Omega(L)$ words of redundancy.
It may also be tempting to store the two data structures with their heads joined in the middle and the tails extending to the two directions respectively.
This can work when the two components can only allocate or release words at tail.
However, the combined data structure needs to allocate or release on both ends.
Hence, one cannot iteratively combine more than two components in this way.

We call such a ``meta-data-structure'' that combines two small data structures into a single large data structure an \emph{adapter}.
We design an adapter with no redundancy such that any word in a small data structure can be accessed in constant time, and each allocation or release only requires $O(\log L)$ words to be relocated.
More specifically, we construct \emph{bijections}, inspired by \emph{consistent hashing}~\cite{KLLPLL97},
\[
  \sigma_{\l_1,\l_2} : \left\{(i,j):i\in\{1,2\}, \, j\in[\l_i]\right\} \rightarrow [L]
\]
for $L=\l_1+\l_2$, such that only $O(\log L)$ elements are matched differently by $\sigma_{\l_1,\l_2}$ and $\sigma_{\l_1,\l_2+1}$ (or $\sigma_{\l_1+1,\l_2}$).\footnote{The technique from \cite{Berger2020MemorylessWA} could also be used to construct such bijections if randomness is allowed, where the number of differently matched elements is $O(\log L)$ \emph{in expectation}. In contrast, we provide a deterministic construction.}
Thus, by storing the $j$-th word of the $i$-th small data structure in the $\sigma_{\l_1,\l_2}(i,j)$-th word in the combined data structure, we only need to relocate $O(\log L)$ words when the sizes change (by one).
By precomputing a lookup table for each $\sigma_{\l_1,\l_2}$ and the differences between adjacent bijections, we can recover the new address of a word and find the words that need relocation efficiently.
Note that since we only need to combine small data structures here (otherwise, the first approach mentioned above already works), the lookup tables are also small to store.
Also note that assuming the two small components only allocate or release a word at tail (the end with a larger index), the combined data structure also only needs to allocate or release at tail.
Thus, we can iteratively apply adapters when there are multiple components.

%% file: prelim.tex
In order to obtain small redundancy, we will combine adapters with the spillover representation~\cite{patrascu2008succincter}.
In this subsection, we give an overview of the spillover representation and the (static) succinct aB-trees.

\paragraph{Spillover representation.}
The spillover representation introduced by \patrascu{}~\cite{patrascu2008succincter} represents a data structure using a pair $(k, m)\in [K]\times\{0,1\}^M$ for integers $K$ and $M$.
Fixing $K$ and $M$, there are $K\cdot 2^M$ such pairs $(k, m)$.
Hence, it is used to represent a data structure of ``$M+\log K$'' bits.
This is often useful for reducing the redundancy incurred due to ``rounding to an integer number of bits'' when a data structure has multiple components.
For a data structure of size $O(n)$, one often chooses $K$ to be $\Theta(n^2)$.
In this case, by using the spillover representation, the redundancy of a component due to rounding (now on $K$) becomes small, as $\log K-\log (K-1)=\Theta(1/K)=\Theta(1/n^2)$.
They add up to only $o(1)$ bits in total.
Also note that for $K=\Theta(n^2)$, the spill $k$ fits in $O(1)$ words, and can be operated on efficiently.
We will apply adapters to a part of $m$ that occupies an integer number of words. 

\paragraph{Augmented B-trees.}
An augmented B-tree is a tree with branching factor $B$ on $n$ leaves, which correspond to the entries of an array $A[1..n]$.
Each node is associated with a label $\phi\in\Phi$. 
The label of leaf $i$ is a function of $A[i]$, and the label of an internal node is a function of the labels of its children.

\cite{patrascu2008succincter} encodes an aB-tree by recursively computing the encoding bottom up using the spillover representation.
For a node $u$ with label $\phi$ rooted at a subtree of size $n$, one constructs an encoding $(k, m)\in[K(n, \phi)]\times \{0,1\}^{M(n, \phi)}$ \emph{conditioned on} $\phi$, i.e., the label $\phi$ will eventually be stored outside this encoding, and will be given when accessing $(k, m)$. 
The range of encoding $[K(n, \phi)]\times \{0,1\}^{M(n, \phi)}$ may also vary based on the label.

Suppose $u$'s children have labels $\phi_1,\ldots,\phi_B$, and we have recursively computed their encoding $(k_i, m_i)\in [K(n/B, \phi_i)]\times \{0,1\}^{M(n/B, \phi_i)}$ for $i=1,\ldots,B$.
Then it was shown that the spills $\{k_i\}$ \emph{together with} the labels $\{\phi_i\}$ can be combined into a single spill $k$ and a few extra bits \emph{conditioned on} $\phi$ (here, one will use the property that $\phi$ is a function of $\phi_1,\ldots,\phi_B$).
By concatenating these extra bits with all $\{m_i\}$, one obtains an encoding $(k, m)$.

The most important feature of this construction is that it incurs \emph{almost no} redundancy.
That is, if for each $i\in \{1,\ldots,B\}$, the number of different encodings $K(n/B, \phi_i)\cdot 2^{M(n/B, \phi_i)}$ is approximately the number of different subarrays of length $n/B$ that will lead to a root label of $\phi_i$, then this also holds for the parent $u$: $K(n, \phi)\cdot 2^{M(n, \phi)}$ is approximately the number of different length-$n$ arrays with root label $\phi$.
Therefore, one is able to store the entire array $A[1..n]$ using nearly-optimal space {conditioned on} the root label.
This is useful, for example, when storing sparse binary arrays.
If we put the number of ones in the subarray in the labels, then the optimal space conditioned on the root label for storing a $n$-bit array with $m$ ones becomes $\log\binom{n}{m}$.
The aB-trees can match this benchmark, instead of $n$ bits, up to two bits of redundancy.
One can also show that this encoding allows us to recursively decode the labels of any root-to-leaf path in constant time per label.

%% file: org.tex
In \cref{sec:adapter}, we first introduce the virtual memory model that facilitates the recursive construction of daB-trees using adapters, and then we present the construction of adapters.
In \cref{sec:daB}, we define daB-trees, and present a weaker result that assumes the updates cannot change the optimal space of any sub-daB-tree by more than $O(1)$ words.
In \cref{sec:remove_assumption}, we present our main result on daB-trees, removing this assumption.
Next, we present applications of daB-trees in \cref{sec:application}.
Finally, we prove a nearly matching lower bound for adapters in \cref{sec:lb}.

%% file: adapter.tex
In this section, we will introduce the most important subroutine in our paper, addressing a fundamental challenge in designing dynamic succinct data structures: maintaining multiple variable-length data structures within contiguous memory space while supporting fast updates and queries. We start by introducing the storage model for variable-length data structures.

\subsection{Virtual Memory Model}
\label{sec:virtual_memory}

We define the \emph{virtual memory model} that will facilitate our recursive construction, which involves variable-length data structures.
Similar to the word RAM model, the memory consists of memory words, each of which stores a $w$-bit binary string. The words are labeled with positive integers $\{1, 2, \ldots\}$, called the \emph{addresses} of the words. One should think that the memory words form a tape that extends to infinity, starting from word $1$.

A variable-length data structure is allowed to use a prefix of memory words on the tape, i.e., words with addresses $[1, L]$ for some natural number $L$. Like in word RAM, we can read or write a memory word given its address $i$ (we may use ``accessing a word'' to refer to reading or writing a memory word).
Additionally, we also allow the data structure to \emph{allocate} or \emph{release} memory words, by increasing or decreasing $L$ by one. 
The memory size $L$ is assumed to be stored outside the data structure, and is always given to the algorithm when performing an operation.

The usable part of the tape is called a \emph{virtual memory} (VM for short), which serves as a variable-size memory to a variable-length data structure.  
As suggested by its name, the VM does not necessarily occupy a consecutive piece of physical memory in the final implementation on a word RAM. 
The addresses will need to be translated between the VM and the physical memory, which may introduce an additional time cost.
Therefore, we will use the following three quantities to measure the time performance of a variable-length data structure when processing an operation.
\begin{enumerate}
\item The number of arithmetic instructions and lookup table queries it performs. Note that lookup tables will always be stored in a static region of the physical memory, which can be accessed with no overhead. Every single instruction of these types takes $O(1)$ (absolute) time.
\item The number of accesses to VM words. The actual time spent will depend on the implementation of the VM.
\item The number of allocations and releases. These types of instructions will be more expensive than the previous ones, and the time cost is also dependent on the implementation of the VM.
\end{enumerate}

\paragraph{Variable-length data structures with spillover representation.} 
A VM can only store an integer number of words, which may cause $O(\wordlen)$ redundancy due to rounding. 
In our application later, it will be combined with the spillover representation, and thus, allowing us to maintain dynamic data structures with subconstant redundancy.

Recall that the spillover representation, introduced in \cref{sec:abtree}, encodes a data structure by a pair $(k, m) \in [K] \times \BK{0, 1}^M$, where $k$ is the spill and $m$ denotes the memory bits. We further define $\l \defeq \midfloor{M / \wordlen}$ and encode the data structure using $(k, m, \incword) \in [K] \times \bk{\BK{0, 1}^{\wordlen}}^{\l} \times \BK{0, 1}^{M - \wordlen \l}$. In this representation, the $M$ memory bits are divided into $\l$ complete $\wordlen$-bit words and an incomplete word of $M - \wordlen \l$ bits. 

As we will see later, such a dynamic data structure may need to request memory allocation or release during updates, i.e., changing $(K, M)$ and consequently $\l$.
The $\l$ complete words, which contain most of the information, are stored under the virtual memory model. Under this model, when $\l$ changes, the creation or deletion of complete words will always happen at the end of the list of complete words.

Similar to \cite{patrascu2008succincter}, the spill $k$ is stored outside, and so is the incomplete word $\incword$. 
The update or query algorithm will always recover them first before accessing the complete words $m$, and rewrite them back after updating $m$.

\subsection{Two-Way Adapters}
\label{sec:two_way_adapters}

Suppose we have two variable-length data structures $D_1$ and $D_2$ with spillover representations. 
Each $D_i$ is divided into $\l_i$ complete words (for $i=1,2$), an incomplete word, and a spill.
Our goal is to store both of them within a contiguous piece of memory.
As we will see in the next section, the two incomplete words and spills can be compressed into only $O(1)$ words, and it turns out that the previous technique in~\cite{patrascu2008succincter} can store them with little redundancy, and allowing constant-time accesses.
Thus, the focus of this section is to store the complete-word parts of the two data structures, each described by a VM of $\l_i$ words respectively, into a single large contiguous piece of memory, which we also model as a large VM of $L \defeq \l_1 + \l_2$ words. We call the two smaller VMs \emph{sub-\!VMs} and call the large, combined VM a \emph{super-\!VM}.

The \emph{two-way adapter} addresses this challenge by maintaining an address translation between the two sub-VMs and the super-VM.
Specifically, we use $(i, j)$ to represent the $j$-th word in the $i$-th sub-VM, while the words in the super-VM are labeled by $\{1, 2, \ldots, L\}$. 
A two-way adapter maintains a bijection $\sigma_{\l_1,\l_2}$ between $A:=\left\{(i,j):i\in\{1,2\},j\in[\l_i]\right\}$ and $B:=[L]$, supporting
\begin{itemize}
    \item \texttt{allocate}($i$): increment $\l_i$ and $L$ by one,
    \item \texttt{release}($i$): decrement $\l_i$ and $L$ by one.
\end{itemize}
The \emph{cost} of a two-way adapter is the number of elements in $A$ that change their images during an allocation or a release, including the last element that was just added or deleted.
Note that since we require the bijection to depend only on $\l_1,\l_2$, we \emph{do not need extra space} to store it.

In the following, we present an adapter with $O(\log L)$ cost.
\begin{lemma}\label{lem:two_way_adapter}
    There is a two-way adapter that has cost $O(\log L)$.
\end{lemma}

Given $A$ and $B$, this bijection is computed by a deterministic matching algorithm that has multiple rounds, inspired by \emph{consistent hashing}~\cite{KLLPLL97}. 
We first apply a ``hash'' function $h$ to map all elements in both sets onto points on a unit circle. 
Then in each round, we match all pairs $(a, b)\in A\times B$ such that $b$ is the \emph{next point on the circle} in the clockwise order after $a$.
All remaining elements proceed to the next round, until all elements have been matched.
See \cref{alg:adapter_matching}.

\begin{algorithm}[!ht]
    \caption{Matching algorithm}
    \label{alg:adapter_matching}
    \DontPrintSemicolon
    
    \SetKwFunction{fmatch}{Match}
    \SetKwProg{Fn}{Function}{:}{}
    
    \Fn(){\fmatch{$A$, $B$}} {
        $A_0 \gets A$\;
        $B_0 \gets B$\;
        $i\gets 0$\;
        \While{$A_i\neq \emptyset$} {
            $X\gets \{(a, b)\in A_i\times B_i: h(b)\textrm{ is the next point clockwise after }h(a)\}$\;
            \For{$(a, b)\in X$}{
                Set $\sigma_{\l_1,\l_2}(a) \gets b$ in the bijection\;
            }
            $A_{i+1}\gets A_i\setminus \{a:\exists b, (a,b)\in X\}$\;
            $B_{i+1}\gets B_i\setminus \{b:\exists a, (a,b)\in X\}$\;
            $i\gets i+1$
        }
    }
\end{algorithm}

To define the hash function $h$, we first represent each element in $A$ as an integer.
Let $\Lmax$ be a fixed parameter such that $\ell_1$ and $\ell_2$ are always at most $\Lmax$, which also means that $B \subseteq [1, 2\Lmax]$.
We represent $(1, j)$ as $3\Lmax - j + 1$ and $(2, j)$ as $3\Lmax + j$.
This representation makes $A$ a set of consecutive integers, i.e., the set of all the integers within the range of $(3\Lmax - \l_1, 3\Lmax + \l_2]$.
In particular, $A$ are represented as integers in $(2\Lmax, 4\Lmax]$, disjoint from $B$.

Next, we will construct a function $h:\mathbb{N}\rightarrow[0, 1)$ (one should think the endpoints of $[0, 1)$ are connected so that $h$ maps non-negative integers to a unit circle).
For any $x \in \mathbb{N}$, we can represent it in binary as $x = \bk{x_k x_{k-1} \cdots x_0}_2 \defeq \sum_{i=0}^k 2^{i}x_i$. Then, the hash value of $x$ is defined as $h(x) \defeq \bk{0.x_0 x_1 \cdots x_{k}}_2 = \sum_{i=0}^k 2^{-i-1}x_i$, i.e., reversing the bits of $x$ and then putting them after the binary point. 
Intuitively, this hash function maps a set of consecutive integers nearly uniformly to the unit circle.
For example, it maps the set $\midBk{2^t}$ onto the $2^t$ equidistant points of the circle for any $t$.
Note that the choice of $k$ in the binary representation of $x$ is not important -- leading zeros in $x$ do not change $h(x)$.
Also note that one does not have to know $\Lmax$ in advance.
Setting $\Lmax=2^t$ for any sufficiently large $t$ does not change the relative order of all images, thus, increasing $t$ as needed does not change the bijection produced by the matching algorithm.
It was introduced for the simplicity of analysis, and one may also view it as a tie-breaker.

The upper bound on the cost of this adapter is proved in two steps.
We first show that for any $\l_1,\l_2$, \cref{alg:adapter_matching} terminates in $O(\log L)$ rounds (\cref{lm:derandomized_num_rounds}).
Then we prove that the cost is bounded by the number of rounds \cref{alg:adapter_matching} runs, up to a constant factor (\cref{lem:dyn_induction}).
Thus, \cref{lem:two_way_adapter} is a direct corollary of these two lemmas.

\begin{lemma}
    \label{lm:derandomized_num_rounds}
    Fix $\Lmax > 0$. For any non-negative integers $\l_1, \l_2 \le \Lmax$ and $L = \l_1 + \l_2$, if we run \cref{alg:adapter_matching} on $A = (3\Lmax - \l_1, 3\Lmax + \l_2] \cap \mathbb{Z}$ and $B = [1, L] \cap \mathbb{Z}$ with function $h$ defined as above, then the algorithm will terminate within $O(\log L)$ rounds.
\end{lemma}

The process of matching elements in $A$ and $B$ is similar to matching a parenthesis sequence in the following sense. 
We view the elements in $A$ as opening parentheses, elements in $B$ as closing parentheses, cut the unit circle at some point $y \in [0,1)$, and turn it into an interval, thereby obtaining a parenthesis sequence.
The matching algorithm can be thought of as matching this parenthesis sequence: we match all the adjacent opening parentheses and closing parentheses in the first round, remove them from the sequence, and then match the remaining parentheses in the second round and remove them, and so on.
It should be noted that the parenthesis sequence obtained in this way may not be \emph{valid}, i.e., there may be a prefix where there are more closing parentheses than opening parentheses, making it impossible to match them perfectly. 
However, by carefully choosing the point $y$ to cut, we can ensure the validity of the parenthesis sequence, which will be explained below.
In this case, the number of rounds is simply the \emph{depth} of this parenthesis sequence.
This structure allows us to prove the lemma.
\begin{proof}
    We first cut the circle at the origin, turning it into an interval $[0,1)$, and obtain a (maybe not valid) parentheses sequence. Consider the excess function $\gamma(y)$ of this parenthesis sequence, which outputs the number of opening parentheses minus the number of closing parentheses in any prefix of this sequence. Formally, for each $y\in [0,1]$, $\gamma(y) \defeq \alpha(y) - \beta(y)$, where $\alpha(y)$ and $\beta(y)$ counts the number of opening parentheses and closing parentheses in $[0,y)$, respectively. As this parenthesis sequence comes from a circle with an equal number of opening and closing parentheses, we have $\gamma(1) = 0$, and we can extend $\gamma(\cdot)$ as a periodic function over real numbers. 

    Clearly, the validity of this parenthesis sequence depends on whether $\gamma(y)$ is non-negative. If $\gamma(y)$ is non-negative for any $y\in [0,1]$, the parentheses sequence is already valid. Otherwise, we select a point $y^*$ that \emph{minimizes} $\gamma(y)$, and cut the circle at $y^*$ instead, turning it into an interval $[y^*, 1+y^*)$. By shifting the division point from the origin to $y^*$, the excess function changes from $\gamma(y)$ to $\gamma(y) - \gamma(y^*)$, which is always non-negative due to the choice of $y^*$.

    Now, we get a valid parenthesis sequence by cutting at the point $y^*$, with depth 
    \begin{align*}
        \max_{y \in [0,1]} \bk{\gamma(y) - \gamma(y^*)}.
    \end{align*}
    It remains to bound $\gamma(\cdot)$, i.e., the following claim implies this lemma.

    \begin{claim}
        \label{clm:reverse_bit_hashing_property}
        For any $y \in [0,1)$, $\abs{\gamma(y)} \le O(\log L)$.
    \end{claim}

    To prove the claim, we first consider the simple case where $L$ is a power of two, i.e., $L = 2^k$. In this simple scenario, a crucial observation is that all $2^k$ elements of $A$ are distributed almost evenly around the circle, and the same applies to the elements of $B$. Specifically, between any two adjacent $2^{k}$-equidistant points of $[0,1]$, there is exactly one element of $A$ and one of $B$. Formally, for any $i \in [2^{k}]$,
    \begin{align*}
        \label{eq:hash_equidistant} 
        \alpha\bk{\frac{i+1}{2^k}} - \alpha\bk{\frac{i}{2^k}} = 1,  \numberthis
    \end{align*}
    and so is $\beta(\cdot)$.

    To see \eqref{eq:hash_equidistant}, we examine how the ``hash'' function $h$ maps elements of $A$ to $[0,1)$. Recall that for any $x \in \mathbb{N}$, $h$ reverses the binary bits of $x$ and places them after the binary point, so that the value $h(x)$ is mostly decided by the last $k$ bits of $x$. Specifically, the integer $x$ can be uniquely written as $x = q \cdot 2^k + r$, where $q\in \mathbb{N}$, and $0 \le r < 2^k$ is the integer formed by the last $k$ bits of $x$. Then, $h(x) = h(r) + 2^{-k}h(q) \in [h(r), \, h(r) + 2^{-k})$, where $h(r)$ is one of the $2^k$ equidistant points on the circle. When $x$ goes over $A$ (an interval of length $2^k$), $r$ goes over $\midBk{2^k}$, and $h(r)$ goes over all the $2^k$ equidistant points of the circle. Hence for any $i \in \midBk{2^k}$, there is exactly one element $x\in A$ such that $x$ is placed in the interval $[ i/2^{k},\,  (i+1)/2^{k})$, which proves \eqref{eq:hash_equidistant}.

    As an immediate corollary of \eqref{eq:hash_equidistant}, we can see
    \begin{align*}
        \alpha(y) = \bigfloor{2^k y} \text{ or } \bigfloor{2^ k y} + 1,
    \end{align*}
    and so is $\beta(\cdot)$. Hence, when $L$ is a power of two, $\abs{\gamma(y)} \le 1$.

    To handle the general case where $L$ is not a power of two, we can first partition both $A$ and $B$ into $O(\log L)$ smaller intervals of lengths powers of two.
    Suppose there are $s$ bits that are one in the binary representation of $L$, then we can write $L$ as a sum of $s$ powers of two, and correspondingly partition $A$ and $B$ into $s$ small intervals.
    Similarly, we can define the quantities $\BK{\alpha_i(y)}$, $\BK{\beta_i(y)}$, and $\BK{\gamma_i(y)}$ for these subintervals, where $\alpha_i(y)$, $\beta_i(y)$, and $\gamma_i(y)$ are defined for the subintervals with the $i$-th smallest size. By applying the result of the simple case above, we get $\abs{\gamma_i(y)} \le 1$ for each $i$.
    Finally, by observing that $\gamma(y)=\sum_{i=1}^s\gamma_i(y)$, we prove $\abs{\gamma(y)} \leq O(\log L)$.
    This implies the lemma.
\end{proof}

We proceed to prove that the cost of the algorithm is bounded by the number of rounds.

\begin{lemma}
  \label{lem:dyn_induction}
  Let $(A', B')$ be a pair of intervals obtained by applying one allocation or release to $(A, B)$, and $T$ be the number of rounds \cref{alg:adapter_matching} runs on $(A,B)$.
  Then \cref{alg:adapter_matching} matches at most $O(T)$ elements differently on $(A', B')$ compared to $(A, B)$.
\end{lemma}

\begin{proof}
  We will prove by induction that, for all $i\geq 0$, the sets obtained in round $i$ of \cref{alg:adapter_matching} $(A'_i, B'_i)$ can differ from $(A_i, B_i)$ by at most two elements: Since $\left|A'_i\right|=\left|B'_i\right|$ and $\left|A_i\right|=\left|B_i\right|$, either
  \begin{enumerate}[label=(\alph*)]
  \item $A_i'=A_i\cup \{a_i\}$ and $B_i'=B_i\cup \{b_i\}$, or
  \item $A_i=A_i'\cup \{a_i\}$ and $B_i=B_i'\cup \{b_i\}$, or
  \item one of $A'_i$ and $B_i'$ is the same as $A_i$ or $B_i$, while the other has one extra element and one missing element, or
  \item $A_i=A_i'$ and $B_i=B_i'$.
  \end{enumerate}

  For $i = 0$, the claim holds since $(A',B')$ is obtained from $(A,B)$ by applying one allocation or release.
  Below, we assume the claim holds for some $i$, and prove it for $i + 1$.

  For case $(a)$, if $a_i$ matches to $b_i$, then in round $i+1$, $A_{i+1}=A_{i+1}'$ and $B_{i+1}=B_{i+1}'$.
  Otherwise, each of $a_i$ and $b_i$ is in one of the following three cases.
  \begin{itemize}
  \item $a_i$ [resp.~$b_i$] is left unmatched. In this case, $a_i$ [resp.~$b_i$] becomes an extra element in $A_{i+1}'$ [resp.~$B_{i+1}'$] compared to $A_{i+1}$ [resp.~$B_{i+1}$].
  \item $a_i$ [resp.~$b_i$] matches to some $b$ [resp.~$a$] that would have been unmatched without $a_i$ [resp.~$b_i$]. In this case, $B_{i+1}'$ [resp.~$A_{i+1}'$] has one fewer element $b$ [resp.~$a$] compared to $B_{i+1}$ [resp.~$A_{i+1}$].
  \item $a_i$ [resp.~$b_i$] matches to some $b$ [resp.~$a$] that would have been matched to some other $a'$ [resp.~$b'$] without $a_i$. In this case, $a'$ [resp.~$b'$] becomes the extra element in $A_{i+1}'$ [resp.~$B_{i+1}'$] compared to $A_{i+1}$ [resp.~$B_{i+1}$].
  \end{itemize}
  In all three cases, each of $a_i$ and $b_i$ leads to one element difference in round $i+1$.
  Hence, the claim holds for round $i+1$ for case (a).

  By applying the same argument, the claim also holds for cases (b) -- (d) as well.
  Moreover, observe that when $\left|A_i'\right|=\left|B_i'\right|=1$, the algorithm must terminate in one round, so the number of rounds that the algorithm runs on $(A',B')$ is at most $T+1$. As only $O(1)$ elements can change their matching in each round, the lemma thus holds.
\end{proof}

Using a two-way adapter, we can maintain two sub-VMs $V_1,V_2$ of sizes $\l_1,\l_2$ respectively on a large super-VM $V$ of size $L=\l_1+\l_2$:
Given $\l_1,\l_2$, by using the bijection $\sigma_{\l_1,\l_2}$ according to \cref{lem:two_way_adapter}, we store the $j$-th word of $V_i$ in the $\sigma_{\l_1,\l_2}(i,j)$-th word of $V$.
Each time we allocate or release a word in one of the sub-VMs and change $\l_1$ or $\l_2$, we switch to a bijection with new sizes.
\cref{lem:two_way_adapter} guarantees that only $O(\log L)$ words will be stored in different places, incurring $O(\log L)$ word-accesses in $V$ to relocate them.
Furthermore, by precomputing all bijections $\sigma_{\l_1,\l_2}$ for $\l_1,\l_2\leq \Lmax$ and the $O(\log L)$ relocations for all possible allocations and releases, we obtain a lookup table of $O(\Lmax^3)$ words in linear time, thereafter, each word-access to $V_i$ takes $O(1)$ time to find its location in $V$, and each allocation or release takes $O(\log L)$ time to identify the relocations.
We proved the following lemma.

\begin{lemma}\label{lem:two_way_adapter_final}
  We can store two small VMs (sub-\!VMs) $V_1, V_2$ of size at most $\Lmax$ in a large VM (super-\!VM) $V$ via a two-way adapter with no redundancy.
  Moreover, by storing a lookup table of $O(\Lmax^3)$ words which can be computed in linear time, each word-access to a sub-\!VM $V_i$ can be done in constant time followed by one word-access to $V$. Each allocation and release on a $V_i$ can be done with $O(\log L)$ time followed by $O(\log L)$ word-accesses to $V$.
\end{lemma}

%% file: dyn_aBtree.tex
In this section, we define and present dynamic augmented B-trees.
We start by defining dynamic augmented B-trees as follows. 

\begin{definition}
  \label{def:daB_tree}
  A \emph{dynamic augmented B-tree} (\emph{daB-tree} for short) of size $n$ is a data structure maintaining an array $A[1..n]$ of elements from the alphabet $\Sigma$, where $n$ is a power of $B$. Like a normal B-tree, it is a full $B$-ary tree of $\log_B n$ levels with elements of $A$ in its leaves. Additionally:
  \begin{itemize}
  \item Every node is augmented with a label from a set $\Phi$. The label of a leaf is determined by its array element $A[i]$; the label of an internal node $u$ is determined by the labels of its $B$ children and the size of the subtree, i.e., $\phi_u = \labfunc(\phi_{1}, \ldots, \phi_{B}, n_u)$ for some function $\labfunc$, where $\phi_u$ is the label of node $u$, $\phi_1, \ldots, \phi_B$ are the labels of the children of $u$, and $n_u$ is the size of the subtree rooted at $u$ (the size of the subtree is defined as the number of leaves in it).
  \item There is a query algorithm. It starts from the root and repeatedly recurses into a child of the current node. At each step, it decides which child to recurse to by examining the labels of the $B$ children. When a leaf is examined, the algorithm outputs the query answer. It is assumed that the query algorithm spends constant time on each node when the daB-tree is not compressed, i.e., the running time is $O(\log_B n)$ per query.
  \item There is an update algorithm which allows us to modify a single element $A[i]$ at a time. When $A[i]$ is changed to a different element $\sigma \in \Sigma$, its augmented label $\phi$ should also be changed according to the fixed function; so do all $A[i]$'s ancestors. 
  \end{itemize}
\end{definition}

Throughout this paper, we will focus on the case $B = 2$, i.e., all daB-trees are binary trees.

As in \cite{patrascu2008succincter}, we define $\treenum[n, \phi]$ as the number of instances of the array $A[1..n]$ that induces a root label of $\phi$. It can be computed recursively by
\[
  \treenum[n, \phi] = \sum_{\phi_1, \phi_2 \,:\, \labfunc(\phi_1, \phi_2, n) = \phi} \treenum[n/2, \phi_1] \cdot \treenum[n/2, \phi_2].
\]
The goal space usage for a daB-tree is $\log \treenum[n, \phi]$, which is necessary in order to distinguish all $\treenum[n, \phi]$ instances. Here we assume the label $\phi$ is stored outside the daB-tree, and will be recovered before accessing the tree.

In this section, we present a succinct presentation of daB-trees when each update is assumed to change $\log \treenum[n_u, \phi_u]$ by at most $O(1)$ words. We will remove this requirement in the next section.

\begin{theorem}
  \label{thm:dyn_aBtree}
  Suppose there is a constant integer $\beta$ such that $\beta w \ge \log n + \log |\Phi|+ \log |\Sigma| + 100$. We can maintain a daB-tree of $n$ elements, such that:
  \begin{itemize}
  \item The daB-tree is stored within $\log \treenum[n, \phi] + 2$ bits in the virtual memory model, assuming free access to the root label $\phi$.
  \item Each query takes $O(\log^2 n)$ time.
  \item Each update to $A[i]$ takes $O(\log^3 n)$ time, assuming that for every node $u$ lying on the path from $A[i]$ to the root, $\abs{\log \treenum[n_u, \phi_u] - \log \treenum[n_u, \phi'_u]} = O(w)$ holds, where $n_u$ and $\phi_u$ refer to the subtree size and the original label of $u$; $\phi'_u$ denotes the label of $u$ after the update.
  \end{itemize}
  Lookup tables of $O(|\Sigma| + |\Phi|^2 n^3 \log n)$ words are precomputed to support the above operations. These tables only depend on $n$, the daB-tree algorithm, and global randomness; they can be shared between multiple daB-tree instances with the same $n$. These tables can be precomputed in time linear in their total size.
\end{theorem}

Note that the space usage of our daB-tree depends on the root label $\phi$. For instance, when using daB-trees for sparse \textsc{Rank}/\textsc{Select} problem where $m \ll n$ elements in $A$ are 1s, with the label of each node being the number of 1s in the subtree rooted at this node, $\treenum[n, m] = \binom{n}{m}$, thus we need $\log \binom{n}{m} + 2$ bits of memory. As updates are made to the array $A$, the required memory may change, in which case the daB-tree allocates or releases memory words according to the virtual memory model (see \cref{sec:virtual_memory}). To avoid the $w$-bit redundancy resulting from rounding up the representation to an integer number of complete words, we allow an incomplete word at the end of the VM to be stored.

\begin{remark}
  An uncompressed daB-tree can support queries and updates in $O(\log n)$ time. As we compress the daB-tree down to $\log \treenum[n, \phi] + 2$ bits, these operations become slightly slower.
\end{remark}

\begin{remark}
  The $\poly n$ words occupied by the lookup table will not be a bottleneck: Suppose we want to maintain an array $A[1..N]$, we set $n = \poly \log N$ and divide $A$ into $N/n$ subarrays of $n$ elements each, then maintain each subarray using a daB-tree. The lookup table size $\poly n = O(\poly \log N)$ is then negligible.
\end{remark}

\subsection{Proof of Theorem \ref{thm:dyn_aBtree}}

Following the discussion in \cref{sec:virtual_memory}, we represent the daB-tree using three parts: a series of complete words, an incomplete word, and a spill.
Our design of the daB-tree is recursive, using the fact that a subtree of the daB-tree is still a daB-tree of smaller size.
We first construct smaller daB-trees for the children of a node and then combine $B=2$ smaller daB-trees to form a larger daB-tree. 
To combine the complete words from smaller daB-trees which are stored in two separate VMs, we use an adapter introduced in \cref{sec:adapter}. It allows us to store a larger VM (super-VM) while simulating the operations on two smaller VMs (sub-VMs). This super-VM is the main part of the larger, combined daB-tree. The remaining parts, i.e., two incomplete words, two spills, and the labels of two children, all fit in $O(1)$ words. Thus, we will use standard techniques to compress them into the memory while leaving a proper spill. In sum, almost all encoding steps are the same as \cite{patrascu2008succincter} except that we use adapters to combine two VMs instead of simply concatenating them in order. 

Following the high-level description above, we now begin the proof of \cref{thm:dyn_aBtree}.

\begin{proofof}{\cref{thm:dyn_aBtree}}
  The proof is by induction on $n$, which is a power of two. In each step, we aggregate two spillover representations of subtrees of size $n/2$ into one with size $n$, in the same manner as \cite{patrascu2008succincter}. Formally, for a daB-tree of size $n$ with root label $\phi$, we encode it with $M(n, \phi)$ memory bits and a spill $k \in [K(n, \phi)]$. The memory bits are divided into $\l(n, \phi) \defeq \midfloor{M(n, \phi) / w}$ complete words, which are stored in a VM, and an incomplete word. Fix $r$ as a parameter whose value will be decided later. We will inductively show that:
  \begin{itemize}
  \item $K(n, \phi) \le 2r$.
  \item $M(n, \phi) + \log K(n, \phi) \le \log \treenum[n, \phi] + 6 (2n-1) / r$.
  \end{itemize}
  We then set $r = 12n$ (here $n$ is the parameter of the whole daB-tree and $r$ remains fixed during the induction), thus $\beta w \ge \log 2r + \log |\Phi| + \log |\Sigma|$ holds according to the condition of \cref{thm:dyn_aBtree}. 
  For simplicity of the proof, we first focus on the space usage while introducing our encoding method; at the end of the proof we will analyze the time usage and lookup table size.

  \paragraph{Base case.} When $n = 1$, the array maintained by the daB-tree only has one element $A[1]$. The root label $\phi$, which is determined by $A[1]$, is stored outside, so the only task is to store $A[1]$ conditioned on $\phi$. Let $\Sigma_\phi \subseteq \Sigma$ denote the set of elements leading to the label $\phi$, then we only need to store an index in $\Sigma_\phi$. This is done by the standard technique stated below.

  \begin{lemma}[{\cite[Lemma 3]{patrascu2008succincter}}]
    \label{lm:encode_small_set}
    For an arbitrary set $\mathcal{X}$ and integer $r \le |\mathcal{X}|$, we can represent an element of $\mathcal{X}$ by a spillover encoding with a spill universe $K$, where $r \le K \le 2r$, and the redundancy is at most $\frac{2}{r}$ bits.
  \end{lemma}

  By applying this lemma on $\Sigma_\phi$, we can encode an element of it using $\log \treenum[1, \phi] + 2/r$ bits (note that $\treenum[1, \phi] = |\Sigma_\phi|$), where the spill universe is $K(1, \phi) \le 2r$.\footnote{Note that the condition of \cref{lm:encode_small_set}, $r\le |\mathcal{X}|$, may not be satisfied. In this case, we just use spill universe $|\mathcal{X}|$ with 0 memory bits, without any extra encoding.} The induction statement holds for $n=1$.

  For the $M(1, \phi)$ memory bits produced by \cref{lm:encode_small_set}, we cut off the leftmost $\l(1, \phi) \cdot w$ bits, storing them as complete words in the VM, and leave the remaining bits as the incomplete word. The VM will only contain $O(1)$ words (and possibly 0 words) as $w = \Omega(\log|\Sigma|)$.

  Moreover, the encoding and decoding procedures can be implemented efficiently with lookup tables of size $O(|\Sigma| + |\Phi|)$ words -- we only need to store the encoding and decoding mappings, which occupy $O(|\Sigma_\phi|)$ words; taking summation over $\phi$ gives the total space $O(|\Sigma| + |\Phi|)$, as the sets $\midBK{\Sigma_\phi}_{\phi \in \Phi}$ form a partition of the alphabet $\Sigma$.

  When an update is made to $A[1]$, $\phi$ may change. We simply redo the encoding procedure above and rewrite all stored information, causing $O(1)$ word accesses to the VM. It is possible that the new label $\phi$ changes the required number of complete words in the VM, in which case we make $O(1)$ allocations or releases to the VM.

  \paragraph{Induction steps.}

  Assume the induction hypothesis holds for $n/2$, and we are going to prove it for $n$.
  The array $A[1 \ldots n]$ is divided into two parts, each consisting of $n/2$ elements and maintained by a smaller daB-tree. These two smaller daB-trees are used as subtrees of the root node.

  The root label $\phi$ is determined by the labels of its children, denoted by $\phi_1$ and $\phi_2$. According to our induction hypothesis, each of the two subtrees can be encoded within $M(n/2, \phi_i)$ memory bits and a spill in universe $[K(n/2, \phi_i)]$, respectively for $i = 1, 2$. The representation of two subtrees, with their complete words storing in two small VMs (sub-VMs), form the starting point of our encoding procedure. Below, we introduce the encoding procedure step-by-step.

  \begin{figure}[t]
    \centering
    \includegraphics[width = 0.75 \textwidth]{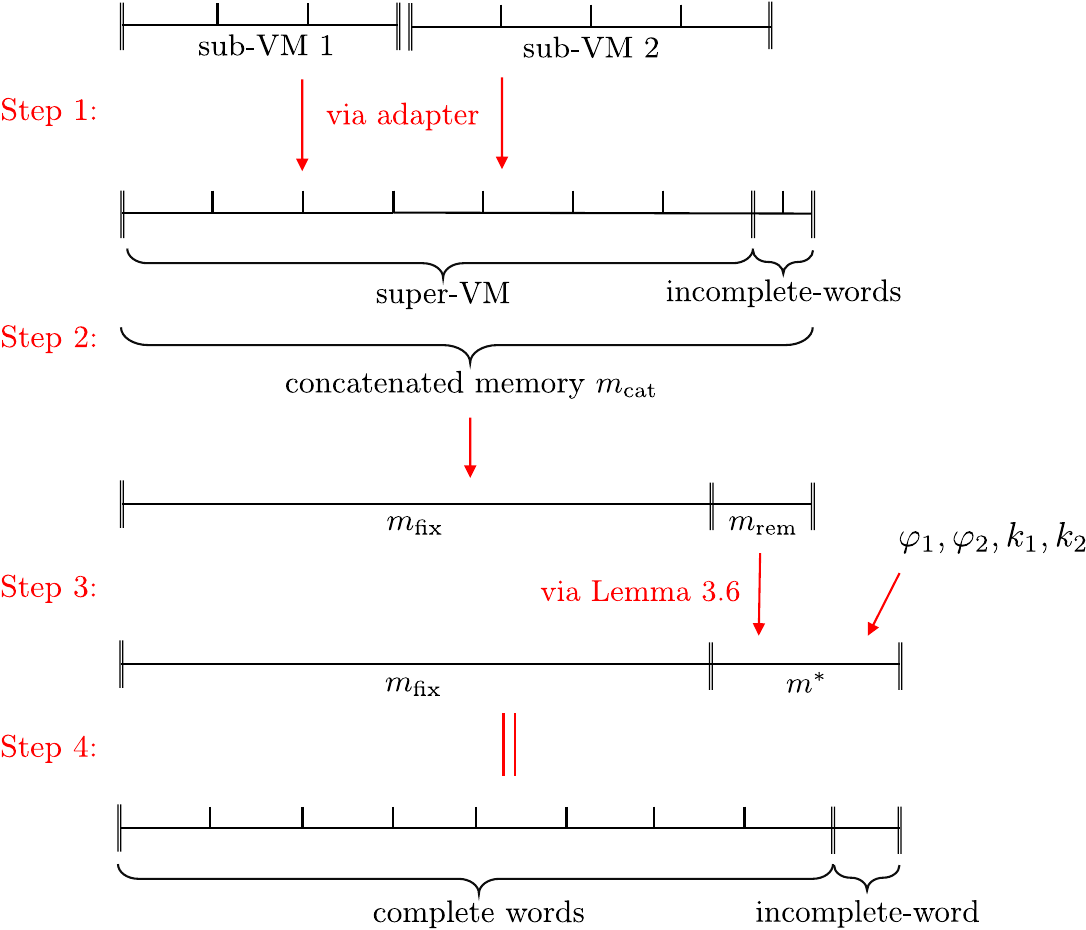}
    \caption{Encoding procedure of 4 induction steps. Step 1 concatenates two sub-VMs using an adapter; Step 2 cuts the concatenated memory into two parts; Step 3 compresses the right part $\mmright$ with labels and spills; Step 4 divides the memory into words and form the final encoding.}
  \end{figure}

  \paragraph{Step 1: Adapter.} As the complete words from the two subtrees are stored in two separate sub-VMs, each containing $\l_1 \defeq \l(n/2, \phi_1)$ and $\l_2 \defeq \l(n/2, \phi_2)$ words, we now use an adapter (\cref{lem:two_way_adapter_final}) to combine them, forming a super-VM of $\l_1 + \l_2$ words. The adapter incurs no redundancy and only requires free access to the lengths of sub-VMs $(\l_1, \l_2)$; we always extract $\l_1, \l_2$ as we recurse to a child, ensuring this requirement is met.

  We view the super-VM as a bit string of length $(\l_1 + \l_2) \cdot w$, and concatenate it with the incomplete words from the two subtrees \emph{at the end}. The outcome is a string of $\Mcat \defeq M(n/2, \phi_1) + M(n/2, \phi_2)$ bits, called the \emph{concatenated memory}, which we denote by $\mcat$. Its length $\Mcat$ depends on the children's labels $\phi_1, \phi_2$, i.e., $\Mcat = \Mcat(\phi_1, \phi_2)$.

  \paragraph{Step 2: Cut the memory.} We cut the concatenated memory $\mcat$ into two parts $\mmleft$ and $\mmright$, such that the first part has $\Mleft$ bits which only depends on the root label $\phi$ but not $\phi_1, \phi_2$; the second part has at most $O(w)$ bits. Formally, we define
  \[
    \Mmax \defeq \max_{\phi'_1, \phi'_2 \,:\, \labfunc(\phi'_1, \phi'_2, n) = \phi} \bk{M(n/2, \phi'_1) + M(n/2, \phi'_2)},
    \qquad
    \Mleft \defeq \max\bk{\Mmax - 4\beta w, \, 0},
  \]
  where recall that $\beta$ is a constant integer satisfying $\beta w \ge \log 2r + \log |\Phi| + \log |\Sigma|$.
  We divide $\mcat$ into the leftmost $\Mleft$ bits and the remaining $\Mright \defeq \Mcat - \Mleft$ bits. It is possible that $\Mcat < \Mleft$ for the current labels $\phi_1$, $\phi_2$, in which case $\mmleft$ is formed by padding zeros to the end of $\mcat$ until it has $\Mleft$ bits; $\mmright$ is left empty. In all cases, the second part contains at most $4\beta w = O(w)$ bits.

  After cutting the memory into two parts, $\mmleft$ directly appears as the leftmost bits in our final encoding, while $\mmright$ is further compressed with other information in the next step.

  \paragraph{Step 3: Compress the labels and spills.} Next, we compress the children's labels $\phi_1, \phi_2$, their spills $k_1, k_2$, and the remaining part $\mmright$ from the last step together, using the following lemma from~\cite{patrascu2008succincter}.

  \begin{lemma}[{\cite[Lemma 5]{patrascu2008succincter}}]
    \label{lm:encode_spillover}
    Assume we need to represent a variable $x \in \mathcal{X}$, and a pair $(\ym, \yk) \in \BK{0,1}^{M(x)} \times [K(x)]$. Let $p(x)$ be a probability density function on $\mathcal{X}$, and $K(\cdot)$, $M(\cdot)$ be non-negative functions on $\mathcal{X}$ satisfying:
    \begin{align*}
      \label{eq:pat_lm5_condition}
      \forall x \in \mathcal{X} :\quad \log \frac{1}{p(x)} + M(x) + \log K(x) \le H. \numberthis
    \end{align*}
    We further assume the word size $w = \Omega(\log |\mathcal{X}| + \log r + \log \max K(x))$, then we can design a spillover representation for $x$, $\ym$, and $\yk$, denoted by $(m^*, k^*)\in \{0,1\}^{M^*}\times [K^*]$, with the following parameters:
    \begin{itemize}
    \item The spill universe is $K^*$ with $K^* \le 2r$; the memory usage is $M^*$ bits.
    \item The redundancy is at most $4/r$ bits, i.e., $M^* + \log K^* \le H + 4/r$.
    \item Given a precomputed table of $O(|\X| \cdot r \cdot \max K(x))$ words that only depends on the input functions $K, M$, and $p$, and assuming $H \le O(w)$, both decoding $(x, \ym, \yk)$ from $(m^*, k^*)$ and encoding $(x, \ym, \yk)$ to $(m^*, k^*)$ takes $O(1)$ time on a word RAM. The table can be precomputed in linear time.
    \end{itemize}
  \end{lemma}

  Let $\kcat \in [K(n/2, \phi_1) \cdot K(n/2, \phi_2)]$ be the combination of the children's spills $k_1$ and $k_2$. We are going to apply the above lemma on $\X = \midBK{(\phi_1, \phi_2) : \labfunc(\phi_1, \phi_2, n) = \phi}$ and $(\ym, \yk) = (\mmright, \kcat)$.

  To construct the probability distribution, we first define
  \[
    p(\phi_1, \phi_2) \defeq \frac{\treenum[n/2, \phi_1] \cdot \treenum[n/2, \phi_2]}{\treenum[n, \phi]},
  \]
  that is, the induced marginal distribution on $(\phi_1, \phi_2)$ if we pick an instance $A[1..n]$ with root label $\phi$ uniformly at random. Same as \cite{patrascu2008succincter}, we slightly perturb the distribution for stronger properties.

  \begin{claim}[\cite{patrascu2008succincter}]
    \label{clm:perturb}
    For any probability distribution $p(\cdot)$ over set $\X$ and any parameter $r > 0$, we can perturb $p(\cdot)$ to another probability distribution $p'(\cdot)$, such that for any $x \in \X$,
    \begin{itemize}
      \item $p'(x) \ge \frac{1}{2r|\X|}$.
      \item $\log \frac{1}{p'(x)} \le \log \frac{1}{p(x)} + \frac{2}{r}$.
    \end{itemize}
  \end{claim}

  We then apply \cref{lm:encode_spillover} with the perturbed distribution $p'$ and
  \[H \defeq \log \treenum[n, \phi] - \Mleft + \frac{12n - 10}{r}.\]
  We check the condition \eqref{eq:pat_lm5_condition} by discussing the following two cases:
  \begin{itemize}
  \item Suppose $\Mcat \ge \Mleft$, i.e., we did not pad zeros to $\Mcat$. In this case, $|\mmright| = \Mcat - \Mleft$. The induction hypothesis implies that
    \[
      M(n/2, \phi_i) + \log K(n/2, \phi_i) \le \log \treenum[n/2, \phi_i] + \frac{6(n-1)}{r} \qquad (i = 1, 2).
    \]
    Therefore, for any $(\phi_1, \phi_2) \in \X$, the left-hand side of \eqref{eq:pat_lm5_condition} is
    \begin{align*}
      & \phantom{{}\le{}} \log \frac{1}{p'(\phi_1, \phi_2)} + \bigbk{\Mcat(\phi_1, \phi_2) - \Mleft} + \log \bigbk{K(n/2, \phi_1) \cdot K(n/2, \phi_2)} \\
      & \le
        \bk{\log \frac{1}{p(\phi_1, \phi_2)} + \frac{2}{r}} + M(n/2, \phi_1) + M(n/2, \phi_2) - \Mleft + \log K(n/2, \phi_1) + \log K(n/2, \phi_2) \\
      & \le
        \log \frac{\treenum[n, \phi]}{\treenum[n/2, \phi_1] \cdot \treenum[n/2, \phi_2]} + 
        \bigbk{M(n/2, \phi_1) + \log K(n/2, \phi_1)} \\
      & \phantom{{} \le \log \frac{\treenum[n, \phi]}{\treenum[n/2, \phi_1] \cdot \treenum[n/2, \phi_2]}} + \bigbk{M(n/2, \phi_2) + \log K(n/2, \phi_2)} + \frac{2}{r} - \Mleft \\
      & \le
        \log \treenum[n, \phi] + 2 \cdot \frac{6(n-1)}{r} + \frac{2}{r} - \Mleft \\
      & = H.
    \end{align*}
  \item Suppose $\Mcat < \Mleft$. In this case, $\mmright$ is empty, so the left-hand side of \eqref{eq:pat_lm5_condition} equals
    \begin{align*}
      & \log \frac{1}{p'(\phi_1, \phi_2)} + \log \bigbk{K(n/2, \phi_1) \cdot K(n/2, \phi_2)} \\
      \le{} \, &\log (2r |\X|) + \log (2r)^2 \,\le\, 3 \log (2r) + 2 \log \midabs{\Phi} \,\le\, 4 \beta w,
    \end{align*}
    where the first inequality is due to \cref{clm:perturb} and the induction hypothesis $K(n/2, \phi_i) \le 2r$; the last inequality holds as $\beta w\ge \log 2r + \log |\Phi| + \log |\Sigma|$. On the other side, the right-hand side of \eqref{eq:pat_lm5_condition} is
    \begin{align*}
      H \; &\ge \max_{(\phi'_1, \phi'_2) \in \X} \log \bigbk{\treenum[n/2, \phi_1'] \cdot \treenum[n/2, \phi_2']} - \Mfix + \frac{12 n - 10}{r} \\
        &\ge \max_{(\phi'_1, \phi'_2) \in \X} \bk{\Mcat(\phi'_1, \phi'_2) - 2 \cdot \frac{6(n - 1)}{r}} - \Mfix + \frac{12 n - 10}{r} \\
        &= \; \Mmax - \Mfix + \frac{2}{r} \;>\; 4\beta w \;\ge\; \textup{left-hand side}.
    \end{align*}
    Here the third inequality is because $\Mfix = \max\midbk{\Mmax - 4\beta w, \, 0}$ and the condition $\Mfix > \Mcat \ge 0$. Therefore \eqref{eq:pat_lm5_condition} also hold in this case.
  \end{itemize}
  Applying \cref{lm:encode_spillover} gives us a spillover representation $(m^*, k^*) \in \midBK{0, 1}^{M^*} \times [K^*]$ of $(\phi_1, \phi_2, k_1, k_2, \mmright)$ conditioning on $\phi$, where $K^* \le 2r$ and
  \[
    M^* + \log K^* \le H + \frac{4}{r} = \log \treenum[n, \phi] + \frac{12 n - 6}{r} - \Mleft.
  \]
  With the help of proper lookup tables, both the encoding and decoding procedures can be completed within constant time.

  \paragraph{Step 4: Concatenate.} The last step involves concatenating $\mmleft$ with $m^*$, the outcome memory bits obtained from the previous step, to form a bit string of length $\Mleft + M^* \eqdef M(n, \phi)$. This memory string, combined with the spill $k^* \in [K^*] = [K(n, \phi)]$, form the encoding for the daB-tree of size $n$. The induction statement holds for $n$, since $K^*\leq 2r$ and
  \[
    M(n, \phi) + \log K(n, \phi) = \Mleft + M^* + \log K^* \le \log \treenum[n, \phi] + \frac{12 n - 6}{r}.
  \]
  Finally, the $M(n, \phi)$ memory bits are again divided into $\floor{M(n, \phi) / w}$ complete words and an incomplete word, while the former ones are stored in a VM.

  \begin{remark}
    Observe that the primary part, $\mmleft$, which is coming from the adapter, is ``aligned'' with the VM. This means that a VM word from the children is still stored as a word in the new VM. This alignment benefits our query and update algorithms because each word-access from a child will translate into only one word-access of the root's VM.
  \end{remark}

  After proving the induction statement for all $n' \le n$ that are powers of two, we pick $r \defeq 12n$ and directly encode the final spill $k \in [K(n, \phi)]$ into memory, incurring a 1-bit redundancy due to rounding. It gives the desired space usage in \cref{thm:dyn_aBtree}:
  \[
    M(n, \phi) + \log K(n, \phi) + 1 \le \log \treenum[n, \phi] + 6(2n-1) / r + 1 < \log \treenum[n, \phi] + 2.
  \]

  \paragraph{Query algorithm and nested adapters.}
  Now we show how to simulate the query algorithm on this succinct representation, in a similar way to~\cite{patrascu2008succincter}.
  The query algorithm starts from the root and walks down to a leaf. 
  Before we visit any node $u$, we already know $u$'s incomplete word, its spill $k = k^*$, and its label $\phi$. We then recover $m^*$, which is stored in $u$'s incomplete word and the rightmost $O(1)$ words of $u$'s VM. According to \cref{lm:encode_spillover}, we can decode $(k^*, m^*)$ to recover the two children's labels $\phi_1, \phi_2$, spills $k_1, k_2$, and the rightmost bits of the concatenated memory, $\mmright$. Consequently, we recover $\l_1, \l_2$, the number of complete words in both children, for the use of adapters. This decoding procedure can be completed within constant time.

  Since we have recovered the labels $\phi_1, \phi_2$ of the children, the query algorithm decides which child to recurse into. Before visiting that child, we first need to recover its incomplete word, which may belong to $\mmleft$, $\mmright$, or both. Since $\mmright$ has been recovered, we only need to make $O(1)$ accesses to $\mmleft$, which is directly stored in $u$'s VM. 
  Then the query algorithm proceeds until it reaches a leaf.

  Note that probing a word in $u$'s VM is not directly allowed for non-root node $u$. The VM of $u$ is connected to its parent's VM via an adapter, and which is further connected to its grandparent via an adapter, and so on. Each time we want to access a word in $u$'s VM, it first translates to an access request to the parent, then to the grandparent, until the root is reached.

  Formally, we use the following subroutine (\cref{alg:probe_word}) to access a word in the VM of an arbitrary node. We denote by $(v, i)$ the $i$-th word in the VM of node $v$ (i.e., the VM that stores the complete words of the representation of the subtree rooted at $v$).

  \begin{algorithm}[!h]
    \caption{Accessing Virtual Memory Words}
    \label{alg:probe_word}
    \DontPrintSemicolon

    \SetKwFunction{fprobe}{Access}
    \SetKwProg{Fn}{Function}{:}{}

    \Fn(\Comment{Access the $i$-th word in $v$'s VM}){\fprobe{$v$, $i$}} {
      \If{$v$ is the root} {
        Directly access the $i$-th word\;
        \Return
      }
      $u \gets v$'s parent\;
      $j \gets$ the translated address of $(v, i)$ in $u$'s super-VM (and also the \emph{concatenated memory})\hspace{-3em}\; \label{line:translate}
      \If{any part of word $j$ belongs to $\mmleft$ of $u$} {
        \fprobe{$u, j$}\;
      }
      \If{any part of word $j$ belongs to $\mmright$ of $u$} {
        Access the corresponding bits in $\mmright$ which was decoded according to \cref{lm:encode_spillover}\;
        \label{line:decode}
      }
    }
  \end{algorithm}

  Here are a few points to mention. In \cref{line:translate}, the translated address $j$ is an address over the super-VM in Step 1 (a prefix of the \emph{concatenated memory}) rather than the VM of $u$. Its major part $\mmleft$ is directly stored in the VM of $u$ while the remaining few bits $\mmright$ are compressed again in Step 3. \cref{line:translate} is done by querying the lookup table of the adapter at $u$ once, which also depends on $\l_1, \l_2$, the lengths of VMs of $v$ and its sibling.

  When we read a word $(v, i)$, \cref{line:decode} does not produce extra word-accesses, because we have already gained the knowledge of $\mmright$ when we visit $u$ (we always need to visit the parent $u$ before we can visit $v$). On the other hand, when we write to some word $(v, i)$ (not required for queries, but required for updates), it might seem that \cref{line:decode} writes multiple words into $u$'s VM; what we do here is to postpone the writing to $\mmright$ until the end of the entire update operation on the daB-tree, at which point we will update all $\mmright$ lying on the path we visited. Thus, reading or writing a word in any node's VM takes $O(\log n)$ time.

  Since the query algorithm initiates $O(1)$ word accesses at every node it visits, the total running time for the query algorithm is $O(\log^2 n)$.

  \paragraph{Update algorithm.} In our update algorithm, we first follow the same procedure as the query algorithm, walking from the root down to the leaf $A[i]$ that needs to update, recovering all labels, spills, and incomplete words of $A[i]$'s ancestors and their siblings.

  As we modify $A[i]$, all its ancestors may need to change their labels. If for some node $u$, the number of complete words $\l(n_u, \phi_u)$ changes after the update, then we request allocations or releases to the adapter connecting $u$ and its parent. Since we have assumed that the change of $\l(n_u, \phi_u)$ is $O(1)$ for every node $u$, the total number of allocations and releases is bounded by $O(\log n)$.

  When we allocate or release a word in a node $u$'s VM, the adapter requires us to change the address mapping for $O(\log (\l_1 + \l_2)) = O(\log n)$ words, which leads to $O(\log n)$ word-accesses in $u$'s parent's VM.
  Each of these word-accesses takes $O(\log n)$ time.
  Therefore, initiating an allocation or a release on any node's VM takes $O(\log^2 n)$ time.
  Thus, it takes $O(\log^3 n)$ time to handle all allocations and releases initiated by an update.
  Note that when we allocate a word in $u$'s VM, the adapter also requires the super-VM of $u$'s parent to allocate a word, which may or may not result in a change of the final number of complete words in $u$'s parent.

  Finally, for all nodes $u$ on the path from bottom to top, we redo all encoding steps introduced above, and update all changed words. They may include:
  \begin{itemize}
  \item During step 1, if $\l_1$ or $\l_2$ changes by $O(1)$, we need to adjust according to the address mapping by moving $O(\log n)$ words, which results in no more than $O(\log n)$ word-accesses at every level.
  \item Suppose the label $\phi_u$ of node $u$ is changed. This may cause $\Mfix$ to change since it depends on $\phi_u$. $\Mfix$ will change by at most $O(w)$ since it only differs from $M(n_u, \phi_u)$ by $O(w)$ while the latter only changes by $O(w)$ during an update. The change of $\Mfix$ requires us to update the rightmost words in $\mmleft$, incurring no more than $O(1)$ word-accesses at each level.
  \item Step 3 involves compressing $O(1)$ words of information into a spillover representation. Its resulting memory bits are directly stored in the rightmost words of $u$'s VM. We initiate $O(1)$ word-accesses to rewrite all of them.
  \end{itemize}
  There are $O(\log n)$ initiated word-accesses at each of the $O(\log n)$ levels, so the total time complexity of the update algorithm is $O(\log^3 n)$.

  \paragraph{Lookup tables.} Our design above involves the following lookup tables:
  \begin{itemize}
  \item Tables for adapters. The number of complete words of the super-VM is bounded by $O(n)$ because $\log \treenum[n, \phi] \le n \log \abs{\Phi} = O(nw)$. That is, the maximum number of words in an adapter is $O(n)$. According to \cref{lem:two_way_adapter_final}, the lookup table consists of $O(n^3)$ words and can be computed in linear time. All adapters in our design can share the same lookup table. 
  \item Tables for the base case $n = 1$. It occupies only $O(|\Phi| + |\Sigma|)$ words and can be precomputed in linear time.
  \item Tables for efficient encoding and decoding in \cref{lm:encode_spillover}. For every label $\phi$ and every $n'\leq n$ that is a power of two, we need a table occupying $O(|\X| \cdot r \cdot \max K(x))$ words.
  Here, $\X$ is the set of pairs $(\phi_1, \phi_2)$ mapping to the parent's label $\phi$, i.e., $\X \defeq \midBK{(\phi_1, \phi_2) : \labfunc(\phi_1, \phi_2, n') = \phi}$. For a fixed $u$ and different $\phi$, all the sets $\midBK{\X_\phi}_{\phi \in \Phi}$ form a partition of $\Phi^2$, so their total size is $O(|\Phi|^2)$; we have $r = O(n)$ and $\max K(x) = \max_{\phi_1, \phi_2} K(n/2, \phi_1) \cdot K(n/2, \phi_2) = O(r^2)$; moreover, there are $O(\log n)$ different $n'$. Multiplying all factors together, we know that lookup tables of this type occupy $O(|\Phi|^2 \cdot n^3\log n)$ words in total.
  \item Tables storing $K(n', \phi)$, $M(n', \phi)$, and $\Mfix = \Mfix(n', \phi)$. They occupy $O(|\Phi| \cdot \log n)$ words and can be precomputed in $O(|\Phi|^2 \cdot \log n)$ time.
  \end{itemize}
  Adding them together, we see that the requirements of lookup table in \cref{thm:dyn_aBtree} are met. Then we conclude our proof.
\end{proofof}

%% file: remove_assumption.tex
We have introduced an approach to succinctly encode daB-trees (\cref{thm:dyn_aBtree}) that works only when $\log \treenum[n_u, \phi_u]$ does not change by more than $O(1)$ words while we modify the label $\phi_u$ of some node $u$. Although this holds for many natural applications including \textsc{Rank/Select}, it is not always the case. Consider the following example: Assume each leaf element is an integer in $[1, 100]$, and we define $\phi_u$ to be the maximum integer within $u$'s subtree. Starting with array $A = (1, 1, \ldots, 1)$ with root label $\phi = 1$, updating an arbitrary element to $100$ will change the root label to $\phi = 100$. In this example, $\log \treenum[n, \phi]$ changes dramatically from $0$ to $O(n)$. In the statement of \cref{thm:dyn_aBtree}, our encoding of the daB-tree is stored in a VM of roughly $\log \treenum[n, \phi]$ bits, so there is no way to avoid $O(n)$ allocations during the operation above, as long as our coding length for the daB-tree only depends on $n$ and the root label $\phi$.

In the proof of \cref{thm:dyn_aBtree}, this unbounded time usage in our encoding algorithm is due to an unusual fact: In the second step, we pad zeros to the end of $\mcat$ if its length is less than $\Mfix$. This padding leads to a situation where $M(n, \phi)$ is much larger than $M(n/2, \phi_1) + M(n/2, \phi_2)$ for certain pairs $(\phi_1, \phi_2)$. When the node's label $\phi$ changes, $\Mfix$ might also change significantly since it depends on $\phi$. As a result, we need to make a dramatic adjustment to the number of padded zeros, leading to an unacceptably large number of allocations in the VM of that node.

Fortunately, we can avoid it via the following approach. Assume we were going to pad many zeros, say at least $(6 \beta w + 1) \log n$ bits (recall that $\beta$ is a constant satisfying $\beta w \ge \log n + \log |\Phi| + \log |\Sigma| + 100$). Equivalently, $\Mcat \le \Mfix - (6 \beta w + 1) \log n$. Instead of padding zeros, we use a different encoding in this node: Besides the concatenated memory $\mcat$ of two children, we directly encode children's labels $\phi_1, \phi_2$, their spills $k_1, k_2$, and other metadata of the children (explained later) within $6 \beta w$ bits, and then attach it to the end of the concatenated memory. There is no spill in this encoding. We call this alternative encoding the \emph{relaxed encoding scheme}, and call the original encoding the \emph{succinct encoding scheme}.

Unlike the succinct encoding scheme, when some node $u$ is adopting the relaxed scheme, its memory length does not only depend on the subtree size $\curn$ and label $\phi_u$. Denote by $\planM$ the number of memory bits in its representation. It should be maintained outside the subtree of $u$, and should be recovered before we can do any operation on $u$.

When any node $u$ is adopting the relaxed scheme, we enforce all its ancestors to also adopt the relaxed scheme. In the whole daB-tree, the nodes adopting the succinct scheme form several disjoint subtrees, for which we can inherit \cref{thm:dyn_aBtree}'s induction statement:
\[
  M(\curn, \phi_u) + \log K(\curn, \phi_u) \le \log \treenum[\curn, \phi_u] + \frac{6(2\curn - 1)}{r},
\]
where $r \defeq 12n$ is a fixed parameter. For any other node $u$ that adopts the relaxed scheme, we need to additionally store which encoding scheme each of its children is using, and if any child is using the relaxed one, node $u$ also needs to store the memory size of its children. These additional information are stored within the rightmost $6 \beta w$ bits, together with the children's labels and spills.

Following the discussion above, we now prove the following variant of \cref{thm:dyn_aBtree}, removing the assumption that $\log \treenum[\curn, \phi_u]$ must only change by $O(w)$:
\begin{theorem}
  \label{clm:dyn_aBtree}
  Suppose $\beta$ is a constant integer such that $\beta w \ge \log n + \log |\Phi| + \log |\Sigma| + 100$. We can maintain a daB-tree of $n$ elements, such that:
  \begin{itemize}
  \item The daB-tree is stored within $M \le \log \treenum[n, \phi] + 3$ bits under the virtual memory model, assuming free access to the root label $\phi$ and the memory size $M$. In particular, $M$ is stored outside, and the daB-tree can update $M$ within constant time.
  \item Each query takes $O(\log^2 n)$ time, and each update takes $O(\log^5 n)$ time.
  \end{itemize}
  Lookup tables of $O(|\Sigma| + |\Phi|^2 n^3 \log n)$ words are precomputed to support the above operations. These tables can be precomputed in time linear in their total size. 
\end{theorem}

\begin{proof}
  Our encoding strategy is as follows.
  For leaf nodes, we use the same encoding scheme as in \cref{thm:dyn_aBtree}.
  Suppose we are going to encode a subtree of size $\curn$ rooted at $u$ with root label $\phi_u$.  We first encode both subtrees rooted at the children of $u$ recursively. If either of the two children is encoded using the relaxed scheme, then we also choose the relaxed scheme for the root. Otherwise, we compare $\Mcat$ and $\Mfix$: If $\Mcat < \Mfix - (6 \beta w + 1) \log n$, we adopt the relaxed scheme, setting its memory size to $\planM = \Mcat + 6 \beta w$; otherwise, we adopt the succinct scheme.\footnote{Note that we will adopt the relaxed scheme only when there are at least $(6 \beta w + 1) \log n$ padded zeros, but only $6 \beta w$ bits of them are used in $\planM$. It seems a waste of space, but we need to make room for all the ancestors of this node, which are forced to adopt the relaxed scheme.}

  The succinct scheme is exactly the same as in \cref{thm:dyn_aBtree}.
  The relaxed scheme, as discussed above, consists of two parts concatenated together.
  The first part $\mcat$ is the concatenation of memory bits from two children using an adapter as in Step 1 of \cref{thm:dyn_aBtree}. The second part is precisely $6 \beta w$ bits, storing children's spills $k_1, k_2$, labels $\phi_1, \phi_2$, their adopted encoding schemes, and memory sizes if relaxed schemes are used.

  We prove the following statement inductively: For any $\curn \le n$ that is a power of two, and for any root label $\phi_u$, letting $r \defeq 12n$ be a fixed parameter, we can encode the sub-daB-tree rooted at $u$ with size $n_u$ and root label $\phi_u$, such that one of the following holds:
  \begin{itemize}
  \item The root $u$ adopts the succinct scheme. The representation consists of $M(\curn, \phi_u)$ memory bits and a spill in universe $K(\curn, \phi_u)$, where $K(\curn, \phi_u) \le 2r$, and
    \[
      M(\curn, \phi_u) + \log K(\curn, \phi_u) \le \log \treenum[\curn, \phi_u] + \frac{6(2\curn - 1)}{r}.
      \numberthis \label{eq:remove_assumption/hypothesis_1}
    \]
  \item The root $u$ adopts the relaxed scheme. The representation involves $\planM$ memory bits without a spill, where
    \[
      \planM \le \log \treenum[\curn, \phi_u] - (6 \beta w + 1) \log \frac{n}{\curn}.
    \]
  \end{itemize}
  Note that in both cases, the memory size of the root $u$ (it equals $M(\curn, \phi_u)$ or $\planM$) is at most $\log \treenum[\curn, \phi_u] + 1$ bits.

  The base case $\curn = 1$ is just the same as in \cref{thm:dyn_aBtree}, since we always choose the succinct encoding scheme for leaves. Next, we assume the statement holds for $\curn / 2$ and we prove it for $\curn$. We denote by $M_1$, $M_2$ the number of memory bits from two children, and $\phi_1, \phi_2$ the labels of two children. There are three cases to consider.

  \smallskip
  \textbf{Case 1.} At least one child adopts the relaxed scheme. Assume the left child is adopting the relaxed scheme.
  For the right child, whichever scheme it chooses, its memory size $M_2$ is at most $\log \treenum[\curn/2, \phi_2] + 1$ due to the induction hypothesis. We further have
  \begin{align*}
    \planM &= \Mcat + 6 \beta w = M_1 + M_2 + 6 \beta w \\
           &\le \bk{\log \treenum[\curn / 2, \phi_1] - (6 \beta w + 1) \log \frac{n}{\curn / 2}} + \bigbk{\log \treenum[\curn / 2, \phi_2] + 1} + 6 \beta w \\
           &\le \log \treenum[\curn, \phi_u] - (6 \beta w + 1) \log \frac{n}{\curn},
  \end{align*}
  so the statement holds.

  \smallskip
  \textbf{Case 2.} Both children adopt the succinct scheme while the root adopts the relaxed scheme. 
  In this case, the condition $\Mcat < \Mfix - (6 \beta w + 1) \log n$ holds. Furthermore, we have $\Mfix  \le \log \treenum[\curn, \phi_u] + 1$,\footnote{One can show $\Mfix \le \Mmax \le \log \treenum[\curn, \phi_u] + 1$ by picking the children's labels $(\phi_1', \phi_2')$ that lead to the root label $\phi_u$ with the maximum concatenated memory $M'_1 + M'_2 = \Mmax$. Then adding up \eqref{eq:remove_assumption/hypothesis_1} for both children implies $\Mmax \le \log \treenum[\curn/2, \phi'_1] \cdot \treenum[\curn/2, \phi'_2] + 1 \le \log \treenum[\curn, \phi_u] + 1$.} and thus
  \[
    \planM = \Mcat + 6\beta w < \Mfix - (6 \beta w + 1) \log n + 6 \beta w\le \log \treenum[\curn, \phi_u] - (6 \beta w + 1) \log \frac{n}{2},
  \]
  which implies the statement.
  
  \smallskip
  \textbf{Case 3.} Both children and the root adopt the succinct encoding scheme. This case is the same as the proof of \cref{thm:dyn_aBtree}.

  \smallskip
  Combining these three cases, we know that the induction statement holds. 
  By storing one bit indicating which encoding scheme we are using for the root, as well as the spill (if using the succinct scheme), the whole daB-tree is stored in a VM of $\log \treenum[n, \phi] + 3$ bits.

  Similar to \cref{thm:dyn_aBtree}, the query algorithm initiates $O(1)$ word-accesses at every level, so the time complexity is $O(\log^2 n)$.

  It remains to bound the time complexity of the update algorithm. 
  Consider how the number of complete words can change in a node during an update.
  For the leaf node, its size is always $O(1)$ words, so it can change by at most $O(1)$ words.
  For an internal node $u$, the child being updated can change its size, resulting in the same changing amount in the size of $\mcat$, possibly with an additional change of $O(1)$ words. Next, we consider how the change of $|\mcat|$ would lead to a change of the encoding length of $u$.
  \begin{itemize}
  \item If $u$ uses the succinct scheme both before and after the update, then the extra encoding from $\mmright$, labels, and spills contains at most $O(1)$ words; the padded zeros at the end of $\mcat$ include at most $O(\log n)$ words. Hence, it can result in at most $O(\log n)$ words of change.
  \item If $u$ uses the relaxed scheme both before and after the update, the extra encoding is precisely $6\beta w$ bits, resulting in no extra changes.
  \item If the scheme changes during this update, then it may result in an extra change of at most $O(\log n)$ words.
  \end{itemize}
  Taking the changes at the descendants of $u$ into account, the number of complete words in any $u$ can change by at most $O(\log^2 n)$.
  That is, the total number of allocations and releases initiated by the algorithm in all levels is at most $O(\log^3 n)$, which costs $O(\log^5 n)$ time in total (as discussed in the proof of \cref{thm:dyn_aBtree}, processing an allocation or a release causes $O(\log n)$ word-accesses, thus takes at most $O(\log^2 n)$ time).

  Moreover, besides the lookup tables in \cref{thm:dyn_aBtree}, no additional lookup table is needed, since for an internal node adopting the relaxed scheme, the last $6\beta w$ bits in its memory can be decoded (or encoded) by constant arithmetic operations. Hence the lookup table is still $O(|\Sigma| + |\Phi|^2 n^3 \log n)$ words, precomputed in linear time.
\end{proof}

%% file: application.tex
Our definition of {dynamic augmented B-tree (daB-tree)} captures a wide class of data structures.
Many problems that require us to maintain an array $A[1..n]$ can be directly solved by a daB-tree.
By further compressing the daB-tree using \cref{clm:dyn_aBtree}, this gives us an efficient succinct data structure.
Below, we sketch several such applications.

\subsection{Sparse Dynamic Rank/Select (Fully Indexable Dictionary)}

Assume we want to maintain an array $A[1..n] \in \midBK{0, 1}^n$, which has at most $m$ ones at any time. Via the following method, we can support \textsc{Rank/Select} queries and single-element updates.

Let $w = \Theta(\log n)$ be the word size. We set $r = 2^{(\log n / \log \log n)^{1/5}}$, and divide array $A$ into $n/r$ subarrays of $r$ elements. For simplicity, we assume $r$ is an integer and $r \mid n$. We call each of these subarrays a \emph{block}. 
For a \textsc{Rank} query, i.e., the sum of a prefix of $A$, the prefix can be divided into two parts: the first $i$ complete blocks, and a prefix of the $(i+1)$-th block. 
A \textsc{Select} query can also be solved by querying among the list of blocks and querying in a single block.
We handle these two parts separately.

Denote by $\phi_i$ the number of ones in the $i$-th block, where $\phi_i\in \Phi=\midBK{0, 1, \ldots, r}$. The first task is to maintain the partial sum of $\phi_i$ over the $n/r$ blocks.
A simple solution is to construct a range tree over the sequence $(\phi_1, \ldots, \phi_{n/r})$. It occupies $O(nw/r)$ bits of space and can support updates in $O(\log n/r) = O(\log n)$ time. However, there is a more efficient implementation given by \cite{patrascu2014dynamic}. Their data structure can maintain an array of $n$ bits, supporting updates or \textsc{Rank/Select} queries in $O(\log n / \log \log n)$ time. Its structure is a $(\poly\log n)$-ary tree, where each subtree maintains its corresponding subarray. During an update or a query, the algorithm walks from the root to some leaf, spending $O(1)$ time on every level.

To adapt their data structure to our demand, we cut off subtrees smaller than $r$ while keeping the upper levels of the tree. For simplicity, we again assume $r$ is a power of the branching factor. There are $O(n/r)$ nodes in the remaining part of the tree, each occupying $O(w)$ space. As a result, we can maintain the partial sum of $(\phi_1, \ldots, \phi_{n/r})$ using $O(nw/r)$ space and $O(\log n/\log w)$ query/update time.

To support \textsc{Rank} queries within some block $i$, we then maintain each block using a daB-tree. The leaves contain $r$ bits of the subarray, while the label $\phi$ is the sum of elements in the subtree (which means the root label is exactly $\phi_i$). By applying \Cref{clm:dyn_aBtree}, every daB-tree is able to solve \textsc{Rank} with $O(\log^5 r)$ time per operation, three bits redundancy, and $O(r^5\log r)$-sized lookup tables.

The remaining task is to concatenate $n/r$ VMs (each storing a daB-tree).
Since each block is sufficiently large, we use the first approach mentioned in \cref{sec:tech}, which is similar to the approach proposed in~\cite{BCDMS99}. The following lemma, if we allow the time to be amortized, is a special case of \cite[Lemma 1]{RR03}.
We prove the non-amortized variant here for completeness.

\begin{lemma}
  \label{lm:chunking}
  Assume there are $B$ VMs of $\l_1, \ldots, \l_B$ bits respectively, where $\l_1, \ldots, \l_B \le L$ and $\sum_{i=1}^B \l_i \le S$ always hold, with parameters $S \le BL$. We can store all these $B$ VMs and their lengths within $S + O(B \sqrt{L w} + Bw)$ memory bits under the Word RAM model with word-size $w = \Omega(\log S)$. Each word-access or allocation/release in any VM takes $O(1)$ time to complete.
\end{lemma}

\begin{proof}
  Let $s$ be a parameter which we will determine later. We divide the memory bits into pieces of size $s$ and call each of them a \emph{chunk}. Then the $i$-th VM is divided into $\ceil{\l_i / s}$ chunks (if $\l_i$ is not a multiple of $s$, we round it up, wasting only $O(Bs)$ bits in total). The total number of chunks from all $B$ VMs is no more than $S / s + B \eqdef \Nchunk$.
  
  We first directly store the lengths of VMs, $\l_1, \ldots, \l_B$, in the physical RAM, which takes $O(B w)$ memory bits.
  
  Let the following $\Nchunk \cdot s$ bits form $\Nchunk$ slots of $s$ bits, where each slot can either store a complete chunk or remain empty. We simply store each chunk in an arbitrary slot, and maintain a pointer of $O(w)$ bits to that chunk. For each VM $i$, we leave the space for $\ceil{L/s}$ pointers for it, even if the actual number of chunks is less than this number. All these pointers occupy $O(w \cdot BL/s)$ memory bits. Once we want to access the $j$-th word in a VM, we first read the $\floor{jw / s}$-th pointer to see where the desired chunk is stored, then go to that slot and access the word.
  
  To support quick allocations, we additionally maintain a list of free slots (i.e., slots without any chunk stored inside). This list occupies $O(\Nchunk \cdot w)$ bits and allows us to pick a free slot for the new chunk. When the rise of $\l_i$ causes a new chunk to appear, we first access the free-slot list, obtaining an arbitrary free slot. We then store the bits in the new chunk into that slot and set up the pointer to it. This procedure can be completed within constant time. When we want to delete a chunk and release its slot, we just delete the pointer and insert that slot into the free-slot list.
  
  In summary, we can store all VMs within
  \begin{align*}
    &\Nchunk \cdot s + O\bk{\frac{BLw}{s}} + O(Bw) + O(\Nchunk \cdot w) \\
    \le{} & S + B\cdot s + O\bk{\frac{BLw}{s}} + O\bk{\frac{Sw}{s}} + O(Bw)\\
    \le{} & S + B\cdot s + O\bk{\frac{BLw}{s}} + O(Bw)
  \end{align*}
  memory bits (the last inequality holds as $S \le BL$). Let $s = \sqrt{Lw}$, the memory usage becomes $S + O(B \sqrt{Lw} + Bw)$, as desired.
\end{proof}

\begin{remark}
  In the allocate-and-free memory model, where we can allocate or free memory blocks of any specified sizes, all $B$ VMs and their lengths can be stored in $\sum_{i=1}^B \ell_i+O(B\sqrt{Lw}+Bw)$ bits. Hence, the final data structure dynamically resizes as the array $A$ gets updated.
\end{remark}

When using the above lemma to store daB-trees of subarrays, the total length of the daB-trees is at most
\[
  \sum_{i=1}^{n/r} \bk{\log \treenum[r, \phi_i] + 3} \le \log \treenum[n, \phi] + 3n/r \eqdef S.
\]
Also, the space usage of a single VM will not exceed $\log \treenum[r, \phi_i] + 3 \le r + 3 \eqdef L$. Substituting these parameters into the lemma above, we know the total space usage for the daB-trees is
\begin{align*}
  S + O\bigbk{B \sqrt{Lw} + Bw} &= \log \treenum[n, \phi] + 3n/r + O\bigbk{n \sqrt{w / r}} = \log \treenum[n, \phi] + O\bigbk{n \sqrt{w / r}} \\
                     &= \log \binom{n}{m} + O\bk{\frac{n \sqrt{w}}{2^{\frac{1}{2}(\log n / \log \log n)^{1/5}}}} \\
                     &= \log \binom{n}{m} + O\bk{\frac{n}{2^{(\log n)^{\frac{1}{5} - o(1)}}}}.
                       \numberthis \label{eq:space_usage_rank}
\end{align*}
Combined with the $O(nw/r) = o(n \sqrt{w/r})$ space usage from the inter-block data structure, the total space usage of our design is still \eqref{eq:space_usage_rank}.

During each update or \textsc{Rank/Select} query, the inter-block data structure takes $O(\log n / \log \log n)$ time, while the daB-tree takes at most $O(\log^5 r) = O(\log n / \log \log n)$ time as well. Our running time is already optimal even if we do not have any constraint on the space usage, shown in the lower bound part of \cite{patrascu2014dynamic}.

\subsection{Dynamic Arithmetic Coding}

Assume we have an array $A$ of length $n$, where each element $A_i$ is in alphabet $\Sigma = \midBK{1, 2, \ldots, |\Sigma|}$ of size $|\Sigma| = O(1)$. For an element $\sigma \in \Sigma$, the number of occurrences of $\sigma$ in $A$ is denoted by $f_\sigma$. We aim to store the array $A$ conditioned on $\midBK{f_\sigma}$, using
\[
  \log \binom{n}{f_1, f_2, \ldots, f_{|\Sigma|}} + R
\]
bits of space, where $R$ is the redundancy. $\midBK{f_\sigma}$ is stored outside; we do not count its space usage. Moreover, we need to support queries and updates to any element $A_i$ efficiently.

To solve this problem, we again let $r$ be a parameter that divides $n$, and split $A$ into $n/r$ blocks. We build a daB-tree for each block, in which every node $u$ is augmented with a label $\phi$ recording the numbers of occurrences of all symbols $\sigma \in \Sigma$. The size of the labels' alphabet is bounded by $|\Phi| \le (r + 1)^{|\Sigma|} = r^{O(1)}$.

Let $f^{[i]}_{\sigma}$ denote the number of occurrences of $\sigma$ within the $i$-th block. Then, the daB-tree for the $i$-th block occupies $\log\, \Bigbinom{r}{f^{[i]}_1, \, \ldots \, , \, f^{[i]}_{|\Sigma|}} + 3$ bits of space.
The total space of daB-trees over all $i$ is bounded by $\log \binom{n}{f_1, \, \ldots\, ,\, f_{|\Sigma|}} + 3n/r$. Besides the daB-trees, the root labels of all blocks are stored directly, which takes another $O(nw / r)$ bits. Finally, we use \Cref{lm:chunking} to combine the (variable-length) VMs that store daB-trees, incurring an additional redundancy of $O\bigbk{n\sqrt{w/r} + nw/r} = O\bigbk{n \sqrt{w/r}}$ bits (we assume $r \ge \log n$, otherwise it is not succinct). The total space usage is therefore
\[
  \log \binom{n}{f_1, f_2, \ldots, f_{|\Sigma|}} + O\bk{\frac{n \sqrt{w}}{\sqrt{r}}}.
\]
The running time for queries and updates are $O(\log^2 r)$ and $O(\log^5 r)$ according to \cref{clm:dyn_aBtree}.

The discussion above gives a time-space trade-off of our algorithm. By choosing $r = \poly \log n$, we achieve $O(n / \poly \log n)$ redundancy for arbitrary $\poly \log n$, while the query and update times are $O(\log^2 \log n)$ and $O(\log^5 \log n)$.

\subsection{Dynamic Sequences}

Another important application is to maintain a sequence of $n$ symbols in $\Sigma = \midBK{1, 2, \ldots, |\Sigma|}$, allowing \textsc{Rank/Select} queries, insertions, and deletions:
\begin{itemize}
\item \textsc{Rank}$(k,\sigma)$: query the number of occurrences of $\sigma$ in the first $k$ entries.
\item \textsc{Select}$(k,\sigma)$: query the location of the $k$-th $\sigma$.
\item \textsc{Insert}($i, \sigma$): insert a new symbol $\sigma$ between original elements $A_{i-1}$ and $A_{i}$.
\item \textsc{Delete}($i$): remove the symbol $A_i$ from the array.
\end{itemize}
Notice that insertions and deletions are more powerful than updates. They also change the length of the array.
Since the dynamic \textsc{Rank/Select} problem has a cell-probe lower bound of $\Omega(\log n / \log \log n)$ query or update time, it is natural to ask what is the smallest redundancy we can achieve under this optimal time complexity. Prior to this paper, the best known approach is \cite{NN14} with $O(n / \log^{1-\varepsilon} n)$ bits of redundancy. In this subsection, we improve it to $O(n\cdot \poly \log \log n / \log^{2} n)$ bits.

Let $s = \log^2 n / \poly \log \log n$ and $r = \poly \log n$ be parameters. We divide the whole sequence $A$ into blocks, with each block containing $\Theta(r)$ symbols. For each block, we construct a daB-tree to maintain the corresponding subarray. However, there is a difference from the previous subsections: Instead of placing a single symbol in each leaf node, we store a subarray of $\Theta(s)$ symbols within every leaf node. In other words, the instance maintained by the daB-tree is a sequence of subarrays of symbols -- it includes not only the information of symbols in $A$, but also how we partition the sequence into leaves. We will show later that the latter part is small.

Unlike in \cref{sec:daB}, a leaf node here cannot be compressed into $O(1)$ words. When any leaf is updated, we redo the encoding process for it and rewrite all its $O(s)$ memory bits. This initiates $O(s/w)$ word-accesses or allocations/releases on the daB-tree, taking a running time of $O(s h^5 / w)$, where $h \defeq \log(r/s) = O(\log \log n)$ is the height of the daB-tree.

The number of symbols stored in a specific leaf can vary up to a constant factor; similar for the internal nodes of the daB-tree. When some node is too unbalanced, we reconstruct its subtree to make it balanced again. Formally, for some internal node $u$, when the number of symbols in its left subtree is at least $1 + 1/h$ times more than that in its right subtree, we immediately reconstruct the subtree rooted at $u$ to make every descendant of $u$ perfectly balanced. On the one hand, the maximum number of stored symbols in a leaf is only $(1+1/h)^h = O(1)$ times larger than the minimum number, so every leaf always stores $\Theta(s)$ symbols. On the other hand, assume the subtree rooted at $u$ has size $n_u$. Starting from a balanced state, there will be at least $n_u / h$ operations before reconstructing $u$'s subtree. The cost of reconstruction is equal to updating $n_u$ leaves. Furthermore, taking into account that the daB-tree has $h$ levels, every insertion/deletion causes $O(h^2)$ leaf updates in amortization due to reconstruction at internal nodes. This results in a running time of $O(sh^7 / w)$.

It is also possible that, after numerous of operations, some block no longer contains $\Theta(r)$ symbols. In this case, we either split a block into two smaller ones, and merge two smaller ones together, or move some symbols to the neighbor block, just like what a normal B-tree does. This ensures that every block contains $\Theta(r)$ symbols at any time.

The encoding method within a leaf is as follows. To encode a subarray of $\Theta(s)$ symbols conditioned on its label on the daB-tree, we use a static aB-tree from \cite{patrascu2008succincter}. The static aB-tree maintains exactly the same labels as our daB-trees. By choosing proper parameters and not rounding up the spill into memory bits, we can make each static aB-tree only incurring $O(s/r) = o(1)$ bits of redundancy. Moreover, every leaf of the static aB-tree stores a fixed number of $\Theta(w)$ symbols, which can be compressed into a single word. For each query, the static aB-tree requests $O(\log (s/w))$ word-accesses to the VM of the leaves of the daB-tree, which takes $O(h \log (s/w)) = O(h \log \log n)$ time to complete.

The remainder is similar to the previous subsections: Additional tree structures are used to answer inter-block \textsc{Rank/Select} queries; \cref{lm:chunking} is used to combine multiple VMs. The label on daB-trees includes the numbers of occurrences of all symbols $\sigma \in \Sigma$ which is suitable for \textsc{Rank/Select} queries.

As discussed above, $\log \treenum[n, \phi]$ for the whole array is larger than the actual entropy of storing the symbols, because an additional partition of $n$ symbols into $\Theta(n/s)$ leaves is stored, which has entropy $\log \binom{n + \Theta(n/s)}{\Theta(n/s)} = \Theta(n \log s / s) = \Theta(n \log \log n / s)$. This is the dominant term in the total redundancy of our data structure, as the other parts only incur little redundancy. The insertion/deletion time is $O(sh^7 / w) + O(\log n / \log \log n)$, where the latter term comes from the optimal inter-block \textsc{Rank/Select} data structure. Let $s = \log^2 n / \log^{8} \log n$, we achieve the optimal insertion/deletion time $O(\log n / \log \log n)$ that matches the cell-probe lower bound of the dynamic \textsc{Rank/Select} problem. The queries have the same running time. The corresponding space redundancy is $O(n \log^9 \log n / \log^2 n)$ bits.

\newcommand{\boxedleftparen}{\tikz[baseline=(char.base)]{\node[shape=rectangle,draw,inner sep=1pt] (char) {$\mkern-2.0mu\texttt{(}$};}}
\newcommand{\boxedrightparen}{\tikz[baseline=(char.base)]{\node[shape=rectangle,draw,inner sep=1pt] (char) {$\texttt{)}\mkern-2.0mu$};}}
\newcommand{\Lpar}{\boxedleftparen}
\newcommand{\Rpar}{\boxedrightparen}

\paragraph{Dynamic Succinct Trees.} \cite{NS14} introduced an algorithm to maintain dynamic rooted trees with little redundancy. Their algorithm is based on maintaining a parentheses sequence $P_1, \ldots, P_{2n}$ that represents the tree, allowing insertions, deletions, and a variety of queries:
\begin{itemize}
\item Finding the position of the parenthesis matching $P_i$.
\item \textsc{Rank/Select} queries on opening or closing parentheses.
\item Finding the position of min/max \emph{excess value} in range $[i, j]$.\\
  $\vdots$
\end{itemize}
Here, the \emph{excess value} of a position $i$ is defined as
\[
  E_i \defeq \sum_{j=1}^i g(P_j), \quad \textup{where} \quad g(P_j) \defeq
  \begin{cases}
    +1, & \textup{if } P_j = \Lpar, \\
    -1, & \textup{if } P_j = \Rpar.
  \end{cases}
\]

Based on the data structure introduced above, for each daB-tree node $u$ which corresponds to a subarray $(P_i, \ldots, P_j)$, we additionally record the following quantities within the label:
\begin{itemize}
\item $E_{j} - E_{i-1} = \sum_{k=i}^j g(P_k)$.
\item $\min_{k \in [i,j]} (E_k - E_{i-1})$, $\max_{k \in [i,j]} (E_k - E_{i-1})$, and the position $k$ that maximizes/minimizes them.
\item The numbers of occurrences of $\Lpar \Rpar$ and $\Rpar \Lpar$, respectively. Also record the first and last parentheses, i.e., $P_i$ and $P_j$, in order to update this information.\footnote{This information is not for the queries listed above, but for several other tree operations in \cite[Table I]{NS14}. E.g., to query the number of leaves within a subarray, we only need to know the number of occurrences of $\Lpar \Rpar$.}
\end{itemize}
Then, our data structure is able to support all types of operations in \cite[Table I]{NS14} with only $O(n \log^9 \log n / \log^2 n)$ bits of redundancy. During any operation, the intra-block time consumption (i.e., on the daB-trees) is bounded by $O(\log n / \log \log n)$. The running time of the inter-block data structure depends on the implementation: If we use a simple range tree, then all operations have $O(\log n)$ time; if we use the $O\bigbk{\sqrt{\log n}}$-ary B-tree presented in \cite[Section 7]{NS14}, many of the operations can be improved to $O(\log n / \log \log n)$ time. In the latter case, our running time of all types of operations are equal to the Variant 1 in \cite[Table I]{NS14}, and with improved redundancy of $O(n \log^9 \log n / \log^2 n)$ bits.

%% file: 2way_adapter_lb.tex
\newcommand{\Cst}{X_{\textup{st}}}
\newcommand{\Xst}{X_{\textup{st}}}
\newcommand{\Ast}{A_{\textup{st}}}
\newcommand{\Bst}{B_{\textup{st}}}
\newcommand{\Cend}{X_{\textup{end}}}
\newcommand{\Xend}{X_{\textup{end}}}
\newcommand{\sample}{\textsf{sample}}
\newcommand{\costu}{\cost_u}

\newcommand{\pst}{p_{\textup{st}}}
\newcommand{\qst}{q_{\textup{st}}}
\newcommand{\pend}{p_{\textup{end}}}
\newcommand{\qend}{q_{\textup{end}}}

In this section, we present a nearly matching lower bound for two-way adapters introduced in \cref{sec:two_way_adapters}. Recall that the cost of the adapters constructed in \cref{lem:two_way_adapter} is $O(\log L)$. We will show a lower bound of $\Omega(\log L / \log \log L)$, showing our construction is nearly optimal, up to a logarithmic factor.

Consider the following balls-to-bins problem: We want to dynamically maintain a bijection from two sets of balls $A \sqcup B = \BK{a_1, a_2, \ldots, a_p} \sqcup \BK{b_1, b_2, \ldots, b_q}$ to a set of bins $\BK{1, 2, \ldots, p+q}$, which indicates how we put each ball into an individual bin. We require the bijection to only depend on $p, q$. An insertion operation increases $p$ or $q$ by one, which creates a new ball and a new bin simultaneously, and might require us to relocate some balls according to the new bijection of $(p + 1, q)$ or $(p, q + 1)$. The cost of an insertion is defined as the number of balls we relocate. Clearly, the above problem is a special case of the 2-way adapter problem (where only allocation is allowed). Below, we consider sequences of $2n$ consecutive insertions consisting of $n$ insertions of $A$-balls and $B$-balls respectively, which transforms the empty starting state $p = q = 0$ to the final state $p = q = n$. We call such a sequence an \emph{instance}.

\begin{theorem}
  \label{thm:lower_bound}
  There is a distribution of instances such that any algorithm for this balls-to-bins problem has an expected amortized cost of $\Omega\bk{\log n / \log \log n}$ on this distribution.
\end{theorem}

The proof of this theorem involves the same framework that appears in \cite{li2023tight}. By Yao's Minimax Principle, we may assume without loss of generality that the algorithm is deterministic.

\begin{proof}
  Fix a parameter $\lambda$ which will be determined later. To start the proof, we build a $\lambda$-ary tree over any instance (sequence of $2n$ insertions): there is a root node representing all $2n$ operations; for every node $u$ representing $2m$ ($m \ge \lambda$) consecutive operations, we divide these operations into $\lambda$ consecutive subsequences, each of $2m/\lambda$ operations and is represented by a child of $u$ (hence $u$ has $\lambda$ children). Nodes representing less than $2\lambda$ operations, which are of depth $h = \floor{\log_{\lambda} n}$, become leaf nodes.\footnote{We can assume $n = \lambda^{h-1} \cdot k$ for an integer $0 < k < r$ without loss of generality, by slightly decreasing the value of $n$ by at most a constant factor, which is negligible in the desired lower bound. Then, each leaf contains exactly $2k$ operations.}

  Based on the tree structure, we can distribute the total cost of the algorithm to each node. Assume ball $x$ is moved in operations $t_1$ and $t_2$ ($t_1 < t_2$) but not between them, consider the two leaves $v_1$ and $v_2$ containing $t_1$ and $t_2$ respectively, then we assign the unit cost of relocating $x$ in the $t_2$-th operation to the \emph{lowest common ancestor} (LCA) of leaves $v_1$ and $v_2$. The cost of any node $u$, written $\cost_u$, is defined as the total number of relocations assigned to it. Clearly, the cost of the algorithm is at least the sum of $\cost_u$ over all nodes $u$ in the tree.

  Next, we will first give lower bounds to $\costu$ based on combinatorial quantities, before we construct the hard distribution accordingly.

  \paragraph*{A lower bound for $\normalfont \cost_u$.} Now we focus on an internal node $u$, which corresponds to an interval of the operation sequence. The $\lambda$ children of $u$ further divide this interval into $\lambda$ sub-intervals, each we call a \emph{segment}. By the previous definition, a ball movement is counted into $\cost_u$ only when it and the previous movement of this ball occur in different segments. From this point of view, it is natural to view all the operations within a segment as a whole when analyzing $\cost_u$.

  We arbitrarily fix the operations before node $u$ and consider two possible subsequences of operations during node $u$, namely $\I_1$ and $\I_2$, with the following properties:
  \begin{itemize}
  \item \emph{Sharing endpoints}: both subsequences transforms the starting state $(\pst, \qst)$ to $(\pend, \qend)$.
  \item \emph{Balanced}: the number of balls inserted to $A$ and $B$ are the same, i.e., $\pend - \pst = \qend - \qst$ which equals half the number of operations.
  \end{itemize}
  Denote the cost of $u$ on instance $\I_i$ by $\cost_u(\I_i)$. The following claim shows a bound of $\cost_u(\I_i)$ for each instance pair $(\I_1, \I_2)$ with the above properties.

  \begin{claim}
    \label{clm:lower_bound_by_dist}
    Let $\dist(\I_1, \I_2)$ be the number of balls among $\{a_{\pst + 1}, a_{\pst + 2}, \ldots, a_{\pend}, b_{\qst + 1}, b_{\qst + 2}, \ldots, b_{\qend}\}$ that are inserted in different segments on $\I_1$ and $\I_2$, then
    \[\cost_u(\I_1) + \cost_u(\I_2) \ge \frac{1}{2}\dist(\I_1, \I_2).\] 
  \end{claim}
  \begin{proof}
    To start the proof, we first consider the special case where $\cost_u(\I_1) = \cost_u(\I_2) = 0$. Basically, we want to show that for each newly inserted ball $x$, the segment $v$ in which $x$ is inserted can be determined by solely $(\pst, \qst)$ and $(\pend, \qend)$, hence it is the same on $\I_1$ and $\I_2$, and $x$ is not counted in $\dist(\I_1, \I_2)$.

    Suppose $x$ is inserted in segment $v$ on $\I_i$. Viewing all the operations within $v$ as a whole, without of loss generality, we can regard them as first putting all the newly inserted balls into all the newly inserted bins arbitrarily, and then applying a permutation over all the balls to obtain the bijection after $v$. Recall that each permutation can be decomposed into several disjoint cycles, i.e., there are a series of balls $x, x_1, x_2, \ldots, x_k$, such that the ball $x$ is moved to the original bin of $x_1$, the ball $x_1$ is moved to the original bin of $x_2$, etc., and finally, the ball $x_k$ is moved to the original bin of $x$ (a newly inserted bin). By the assumption that $\cost_u(\I_i) = 0$, all the balls $x, x_1, \ldots, x_k$ cannot be moved in other segments, and in particular, each ball in this cycle is either a ball in $\Xst \defeq \midBK{a_1, \ldots, a_{\pst}, b_1, \ldots, b_{\qst}}$, or a newly inserted ball in segment $v$. Let $x_s$ be the first newly inserted ball after $x$ in this cycle.
    
    Next, we show that the chain $x, x_1, \ldots, x_{s-1}$ can be determined by solely $(\pst, \qst)$ and $(\pend, \qend)$. As these balls are only moved during the segment $v$, by examining the bin containing $x$ in the ending state $(\pend, \qend)$, and the ball in this bin at the beginning $(\pst, \qst)$, we can determine $x_1$. We can further determine $x_2, \ldots, x_{s-1}$ similarly, until $x_{s-1}$, for which we found that it is in a newly inserted bin of segment $v$ in the ending state. Then the next element $x_s$ will be a newly inserted ball. Hence, both the chain and $v$ can be determined by the beginning and ending states, which are same in $\I_1$ and $\I_2$, implying $\dist(\I_1, \I_2) = 0$.

    Now we extend this argument to the general case where the chains are not completely determined by the starting and ending states, $(\pst, \qst)$ and $(\pend, \qend)$. Suppose ball $x$ was inserted in segment $v$ on $\I_i$, we observe the chain $x, x_1, x_2, \ldots, x_k$ that describes a cyclic movement in the segment $v$ on $\I_i$, which we call $x$'s \emph{induced chain}. If none of $x, x_1, \ldots, x_k$ is moved in another segment (on $\I_i$), this chain could be determined by $(\pst, \qst)$ and $(\pend, \qend)$ in the same way as when $\costu(\I_i) = 0$, as well as the segment $v$ when $x$ was inserted; if this condition further holds on both sequences $\I_1$ and $\I_2$ for the same ball $x$, then the chain and $v$ must be the same on both sequences.
    Otherwise, there is a ball $y$ in the chain $x, x_1, \ldots, x_k$ that is moved in another segment on $\I_i$, in which case we say this additional movement \emph{destroys} the chain, and it is no longer guaranteed to be the same on both sequences. When a ball $y$ is moved in $t$ different segments on $\I_i$, it belongs to (and thus can destroy) $t \le 2(t-1)$ chains while contributing $t-1$ to $\costu(\I_i)$. Thus, $\costu(\I_1) + \costu(\I_2)$ is at least half the number of destroyed chains on both sequences.

    Moreover, if a ball $x$ is inserted in different segments on $\I_1$ and $\I_2$ (thus is counted in $\dist(\I_1, \I_2)$), $x$'s induced chain on at least one of $\I_1$ and $\I_2$ must be destroyed. This implies $\costu(\I_1) + \costu(\I_2) \ge \dist(\I_1, \I_2) / 2$.
  \end{proof}
  
  \paragraph*{Construction of the hard distribution.} Motivated by \cref{clm:lower_bound_by_dist}, we construct \cref{alg:sample_of_instance}, which enlarges $\E[\dist(\I_1, \I_2)]$ for most of the nodes $u$.

  \begin{algorithm}[htbp]
    \captionof{Distribution}{Hard Distribution}
    \label{alg:sample_of_instance}
    \DontPrintSemicolon
    
    \SetKwFunction{sample}{Sample}
    \SetKwProg{Fn}{Function}{:}{}
    
    \Fn(\Comment{Sample the operation sequence of $u$ ($m$ insertions to both $A$ and $B$)\hspace{-2em}}){\sample{$u$}} {
      \uIf(\Comment{$u$ is a leaf node}){$m < \lambda$} {
        Arbitrarily insert $m$ balls to $A$ and $B$ respectively\;
      } \Else {
        $r \gets$ uniformly random from $\BK{0,1}$\;
        \For{each child $v$ of $u$} {
          \uIf{$v$ is the leftmost child} {
            In $v$, arbitrarily insert $2rm/\lambda$ balls to $A$ and $(2 - 2r)m/\lambda$ balls to $B$\;
          } \uElseIf{$v$ is the rightmost child} {
            In $v$, arbitrarily insert $(2 - 2r)m/\lambda$ balls to $A$ and $2rm/\lambda$ balls to $B$\;
          } \Else{
            Run \sample{$v$} to get the operation sequence in $v$\;
          }
        }
      }
    }
    Run \sample{$r$} on the root node $r$
  \end{algorithm}

  \bigskip
  
  Based on \cref{alg:sample_of_instance}, we can bound the expectation of $\cost_u$ for each node $u$ where $\sample(u)$ is called (say $u$ is \emph{activated}), shown in the following claim.

  \begin{claim}
    For each activated node $u$ with $2m$ operations in it, $\E\Bk{\cost_u} \ge m/8$.
  \end{claim}
  \begin{proof}
    We sample two operation sequences in node $u$ according to \cref{alg:sample_of_instance} independently, denoted by $\I_1, \I_2$. We have
    \[
      \E[\costu] = \frac{1}{2} \E\bigBk{\costu(\I_1) + \costu(\I_2)} \ge \frac{1}{4} \E[\dist(\I_1, \I_2)]
    \]
    according to \cref{clm:lower_bound_by_dist}. Then, we consider the random bit $r$ in the run of $\sample(u)$. With $1/2$ probability, $\I_1, \I_2$ have $r = 0, 1$ respectively, or vice versa. It only remains to show that $\dist(\I_1, \I_2) \ge m$ in this case. (Due to symmetry, we only consider when $\I_1$ has $r = 0$ and $\I_2$ has $r = 1$.)

    To compute $\dist(\I_1, \I_2)$, we compare the newly inserted balls to $A$ on $\I_1$ and $\I_2$. We label all the newly inserted balls of $A$ by integers in $[m]$ in the order of insertion. By \cref{alg:sample_of_instance}, for each $1<i<\lambda$, the set of $A$-balls inserted in the $i$-th segment on $\I_1$ is $\big( (i-2)m/\lambda, \, (i-1)m/\lambda \big]$, while that on $\I_2$ is $\big( im/\lambda, \, (i+1)m/\lambda \big]$. Hence, there are $m/\lambda$ balls in $A$ which are inserted in the $i$-th segment on $\I_1$ but in other segment on $\I_2$. Add all the segments together, the balls in $A$ contribute $(\lambda - 2) \cdot m/\lambda \ge m/2$ to $\dist(\I_1, \I_2)$. It is similar for balls in $B$. Therefore, $\dist(\I_1, \I_2) \ge m$.
  \end{proof}
  
  Finally, we can sum up the cost of each node together to bound the total expected cost of the hard distribution.

    Set $\lambda = \log^2 n$, then the height of the tree is $h = \ceil{\log_\lambda n} = \bigceil{\frac{\log n}{2\log \log n}}$. Note that each activated node has $1 - 2/\lambda$ fraction of children being activated. Hence, in the $\l$-th level of the tree from the root (the root is in level 0), there are $\bk{1 - \frac{2}{\lambda}}^\l \ge 1 - \frac{2h}{\lambda} \ge 1 - \frac{2}{\log n \log \log n} \ge \frac{1}{2}$ fraction of nodes to be activated. Taking summation over all the activated nodes, the expected amortized cost over this hard distribution is at least 
    \begin{align*}
        \frac{1}{n} \sum_{\text{node }u} \E \Bk{\cost_u}
        &\ge \frac{1}{n} \sum_{\l = 1}^{h-1} \frac{1}{4} \cdot n_\l \cdot \frac{2n}{n_\l}
          \ge \Omega\bk{\frac{\log n}{\log \log n}}
    \end{align*}
    where $n_\l = \lambda^{\l}$ represents the number of level-$\l$ nodes on the tree.
\end{proof}